\documentclass[11pt]{article}

\usepackage{amsthm}

\usepackage[letterpaper,margin=1in]{geometry}
\usepackage{amssymb}
\usepackage{amsmath}
\usepackage{amstext} %

\usepackage{appendix}
\usepackage[algo2e,ruled,vlined,linesnumbered]{algorithm2e}
\DeclareMathOperator{\rank}{rank}

\usepackage{epstopdf,epsfig}    
\usepackage[sort,nocompress]{cite}
\usepackage{doi}

\usepackage{lipsum}

\newcommand{\Reals}{\mathcal{R}}

\newcommand{\Prob}{\mathbb{P}}

\newcommand{\cB}{\mathcal{B}}      %

\newcommand{\cX}{\mathcal{X}}

\usepackage{array,booktabs}
\usepackage{multirow}
\usepackage{xcolor}

\usepackage[normalem]{ulem} %

\newcommand{\replace}[2]{#2} 
\usepackage{marginnote}
\normalmarginpar

\usepackage{xspace}
\usepackage{hyperref}
\hypersetup{colorlinks=true, linkcolor=blue!80!black, urlcolor=blue!80!black, citecolor=blue!80!black}

\usepackage{thmtools}
\usepackage{nameref}
\usepackage[nameinlink]{cleveref}

\newtheorem{defi}{Definition}

\crefname{defi}{definition}{definitions}
\Crefname{defi}{Definition}{Definitions}
\crefname{lemma}{lemma}{lemmas}
\Crefname{lemma}{Lemma}{Lemmas}
\crefname{assumption}{assumption}{assumptions}
\Crefname{assumption}{Assumption}{Assumptions}

\usepackage{multicol}

\usepackage[square, comma, sort&compress, numbers]{natbib}
\usepackage{etoolbox}
\makeatletter
\patchcmd{\@mn@@@marginnote}{\begingroup}{\begingroup\@twosidefalse}{}{\fail}
\reversemarginpar
\makeatother
\newcommand{\paragraphsummary}[1]{}
\newcommand{\hlabel}{\phantomsection\label}

\crefname{assumption}{assumption}{assumptions}
\Crefname{assumption}{Assumption}{Assumptions}

\crefrangelabelformat{assumption}{#3#1#4--#5#2#6}

\newtheorem{theorem}{Theorem}
\newtheorem{lemma}{Lemma}
\newtheorem{proposition}{Proposition}

\newtheorem{remark}{Remark}

\newtheorem{assumption}{Assumption}

\renewcommand{\baselinestretch}{1.2}

\usepackage{authblk}
\begin{document}

\title{A Novel Noise-Aware Classical Optimizer for Variational Quantum Algorithms} 

\author[a]{Jeffrey Larson}
\author[a]{Matt Menickelly}
\author[b]{Jiahao Shi}

\affil[a]{Mathematics and Computer Science Division, Argonne National Laboratory.  {\tt jmlarson@anl.gov, mmenickelly@anl.gov}}
\affil[b]{Industrial and Operations Engineering Department, University of Michigan. {\tt jiahaos@umich.edu}}
\date{}

\maketitle

\begin{abstract}
A key component of variational quantum algorithms (VQAs) is the choice of classical optimizer employed to update the parameterization of an ansatz. 
It is well recognized that quantum algorithms will, for the foreseeable future, necessarily be run on noisy devices with limited fidelities. 
Thus, the evaluation of an objective function (e.g., the guiding function in the quantum approximate optimization algorithm (QAOA) or the expectation of the electronic Hamiltonian in variational quantum eigensolver (VQE)) 
required by a classical optimizer is subject not only to stochastic error from estimating an expected value but also to error resulting from intermittent hardware noise. 
Model-based derivative-free optimization methods
have emerged as popular choices of a classical optimizer in the noisy VQA setting, based on empirical studies.
However, these optimization methods were not explicitly designed with the consideration of noise. 
In this work we adapt recent developments from the ``noise-aware numerical optimization'' literature to these commonly used derivative-free model-based methods. 
We introduce the key defining characteristics of these novel noise-aware derivative-free model-based methods that separate them from standard model-based methods.
We study an implementation of such noise-aware derivative-free model-based methods and compare its performance on demonstrative VQA simulations to classical solvers packaged in \texttt{scikit-quant}. 
\end{abstract}

\section{Introduction}

\paragraphsummary{Talk about VQA}
Variational quantum algorithms (VQAs) form a class of quantum algorithms and are
a leading method of experimentation for quantum computing researchers and practitioners \citep{Cerezo2021}.
VQAs have quantum circuit representations that are relatively simple and short
and are therefore particularly well suited for near-term quantum computers, which have
limited qubits and are prone to errors \citep{McClean2016,Farhi2014}. 
The ``variational'' nature of VQAs comes from 
the manner in which VQAs combine classical optimization methods with quantum
information techniques. 
The core idea behind VQAs is to use a quantum computer to evaluate a parameterized quantum
circuit; the final quantum state returned by a VQA therefore depends on a given
set of parameters. 
These parameters are adjusted by some classical
numerical optimization algorithm that seeks the optimum of a well-designed and problem-specific cost function. 
\replace{}{Two of the most common VQAs are the quantum approximate optimization algorithm (QAOA)~\cite{Farhi2014} and the variational quantum eigensolver (VQE)~\cite{Peruzzo2014}, which differ in their application and form.
QAOA has shown promise for solving binary optimization problems, for example,
MaxCut variants \citep{Wang2018,Shaydulin2022} and
low autocorrelated binary sequences \citep{Shaydulin2023} problems, where as VQE
has shown promise in performing quantum simulations to identify ground state energies of complex molecules \citep{Kandala2017}.
The efficacy and flexibility of VQAs have piqued the
interest of researchers in chemistry \citep{Grimsley2019,mccaskey2019quantum,yeter2021benchmarking}, materials science \citep{Bauer2020}, and machine learning \citep{Cerezo2021}.}

\paragraphsummary{Talk about DFO in general}
In this manuscript we are particularly interested in the classical numerical optimization method that forms a critical part of a VQA.
While some work suggests exploiting \emph{parameter shifts} \citep{Schuld2019} to compute gradients with respect to the parameters, quantum circuits for computing these gradients are generally very large and hence far more prone to errors resulting from hardware noise. 
Some exploratory work in gradient-based classical optimizers for the VQA setting was performed by \cite{menickelly2023latency, menickelly2023stochastic, kubler2020adaptive, arrasmith2020operator, gu2021adaptive, ito2023latency}. \replace{}{These gradient-based optimizers  offer improved theoretical convergence over derivative-free methods (\cite{Harrow2021})}. 
However, 
in the near term where gate depths are prohibitive, the classical numerical optimization methods used to optimize VQA parameters will likely continue to have access only to cost function values. 
Given the complexity of the cost functions for these quantum systems, the employed optimization methods must be robust and efficient in navigating a landscape that is generally nonconvex and periodic \citep{Shaydulin2019}.
This challenge, central to our motivations in this paper, is compounded by the fact that quantum measurements are inherently probabilistic;
this means that the cost function values used in VQA are generally only statistical estimates.
In particular, one is limited to sampling repeated measurements (called \emph{shots}) of the output state of a quantum circuit; because most cost functions are expectations, a sample average of shot measurements typically yields an unbiased estimate of the cost function. 
However, these estimates necessarily introduce stochastic noise (in addition to, or separate from, any hardware noise) into the optimization process.

\replace{}{\subsection{Existing Methods}}
\paragraphsummary{Stochastic methods should be more appropriate, but everyone uses deterministic methods (for good reason)!}
With the consideration of practically unattainable gradients, practitioners typically use methods for derivative-free optimization (DFO) \citep{LMW2019} as the classical optimizer in a VQA. 
By derivative-free methods we refer to any optimization method that does not require any derivative information to be supplied by the user or the cost function oracle. 
For example, the software package \texttt{scikit-quant} \citep{lavrijsen2020classical} wraps a variety of derivative-free optimization solvers that the authors found to perform well on several benchmark VQE problems.

As discussed, the optimization problem solved on the classical computer is inherently stochastic.
Despite this, the quantum computing community has largely found DFO methods that are explicitly for optimizing stochastic responses to be far less efficient than their counterparts for deterministic responses; for example, none of the solvers wrapped in \texttt{scikit-quant} are intended for stochastic optimization, by the standard definition of stochastic optimization. 
However, as is obvious in theory and is indeed observed in practice, any method designed for a deterministic problem will never resample a cost function at the same parameter setting twice. 
As a result, deterministic methods are certain to eventually ``get stuck''; that is, they will fail to find improvement resulting from small perturbations of the circuit parameters. This ``stuckness'' is likely to happen when the noise becomes relatively large compared with the change in the cost function resulting from a parameter perturbation suggested by the methods. 

\paragraphsummary{Talk methods we are building on. Strength/weakness of Powell's methods}
Some of the most successful optimization methods for \replace{deterministic}{} DFO problems (VQA or otherwise) are slight modifications to the model-based trust-region (MBTR) framework \citep{conn2009introduction}, with notable implementations having been developed by Michael Powell, including the  popular \texttt{BOBYQA}~\citep{powell2009bobyqa}. 
MBTR methods construct and update a quadratic model of the objective function over a dynamically adjusted trust region, typically a norm ball of fixed radius in the parameter space.
MBTR methods allow for efficient approximation of the overall objective function's behavior, thereby guiding the optimization process more reliably, even in the absence of explicit gradient information. The method's quadratic models are interpolatory; these models are especially  useful when the signal from a noisy response is large relative to the noise. 
If the signal-to-noise ratio becomes small, however, the interpolatory models are likely not to model the response but rather provide a model of noise. 
In this low-signal setting, MBTR methods are likely to become stuck.
However, MBTR methods empirically make reasonable progress for as long as the signal-to-noise ratio remains high. 

\replace{}{\subsection{Our Contribution}}
\paragraphsummary{Talk about asymptotic pros/cons for results that we are about to give and how they are complemented by practical considerations.}
In this work we seek to explicitly account for the observation that theoretically, and practically, MBTR methods perform best when the signal-to-noise ratio remains high. 
We accomplish this goal by developing what we call a \emph{noise-aware} MBTR method.
By noise-aware, we broadly refer to any method that effectively requires an \emph{estimate} of the noise level present in the function evaluation. 
We currently leave the meaning of noise level intentionally vague, but for intuition, noise level could refer to any of the following:
\begin{itemize}
\item a deterministic bound on the absolute value of noise, 
\item the standard error when one employs sample means to estimate an expectation value, or
\item the noise level of a deterministically noisy function, as employed in \cite{more2011estimating},
\end{itemize}
among other quantities. 
We note in particular that we do not label our method a \emph{stochastic} MBTR because, for instance, of the three examples presented above, only the second example assumes anything stochastic about the function being estimated. 
In the VQA setting, this choice is motivated by the practical concern that observations of circuit evaluations on a near-term quantum computer are  the result not only of a stochastic calculus but also of hardware noise.  
There is generally no reason to believe that hardware noise satisfies distributional assumptions as pleasant as unbiasedness or parameter independence. 
Moreover, from a theoretical perspective, the analysis of stochastic optimization methods tends to focus on (non)asymptotic convergence rates as a primary concern.
For convergence analyses of various stochastic MBTR methods, see, for example, \citep{stormoriginal,stormrate,Larson2016,Shashaani2018,Chang2013,Augustin2017}. 
All of these methods effectively assume that, in the presence of stochastic noise, one must sample the objective function in each iteration $k$ at a rate like $\mathcal{O}(1/\Delta_k^4)$, where $\Delta_k$ is a trust-region radius. 
Because the theoretical convergence of an MBTR method requires $\Delta_k\to 0$, this sampling rate can quickly become  infeasible in practice.  
As a result, the stochastic optimization methods suggested by these analyses focus on long-term convergence, often at the expense of efficiency in the number of function evaluations spent while satisfying a fixed budget.

\paragraphsummary{What are we doing here? and How}
Therefore, we seek to provide the best possible theoretical convergence results for a \emph{noise-aware} MBTR method. 
We aim to achieve this by building on recent results in what we would call noise-aware optimization. 
Of particular note are \cite{sun2022trust} and \cite{Cao2023}, which provide convergence analyses for trust-region frameworks given explicit access to an estimate of a noise level.
The  definitions of ``noise level'' differ between these two reference works; we believe that the framework in \cite{Cao2023} provides a bit more flexibility that allows the results to be especially appropriate in the VQA setting. 
The results of both \cite{sun2022trust} and \cite{Cao2023} can be characterized as providing a neighborhood of convergence for an optimization method given a reasonable estimate of the noise level.
Thus, instead of trying to achieve arbitrary accuracy (which, for instance, would require an insurmountable number of samples in a stochastic setting), the analysis provides a worst-case rate of convergence to a solution of ``best possible'' accuracy, given the noise level. 
The primary achievement of this manuscript is to take the noise-aware trust-region framework of \cite{Cao2023} and, through careful consideration of interpolation model construction, yield a \emph{practical implementation of a noise-aware MBTR method}. 
The convergence of our method will follow immediately by algorithmically ensuring that the interpolation models satisfy certain properties that are assumed in \cite{Cao2023}.  

\replace{}{\subsection{Paper Organization}}
\paragraphsummary{Paper outline and notation}
The outline of the manuscript is as follows.
\Cref{sec:setting} develops the problem setting, and
\Cref{sec:noisemodel} discusses the chosen noise model suitable for VQAs.
\Cref{sec:algorithm} presents a general noise-aware trust-region algorithm from the literature that we will build upon,  and
\Cref{sec:Preliminary.theory} presents some of its theoretical properties.
\Cref{sec:novelalg} presents our novel method for noise-aware optimization, and 
\Cref{sec:theory} analyzes the convergence properties of its sequence of iterates.
\Cref{sec:results} presents numerical results for our novel algorithm on a selection of noisy optimization problems.
\Cref{sec:discussion} concludes with open questions for further investigation.

\section{Setting}\label{sec:setting}
Given a computational ansatz, let a quantum circuit be parameterized by $\theta\in\Reals^d$. 
We aim to solve the optimization problem
\begin{equation}
\label{eq:basic_prob}
    \text{minimize} f(\theta): \theta\in\Reals^d,
\end{equation}
where $f: \Reals^d\to\Reals$ denotes some objective function (e.g., a guiding function in QAOA or the energy of the Hamiltonian in VQE). %

We first discuss, in \Cref{sec:noisemodel}, a particular choice of noise model from recent literature that appears to be appropriate for the noisy VQA setting. 
A suitable trust-region algorithm for the minimization of objective functions subscribing to this noise model was also recently proposed, and we present it in \Cref{sec:algorithm}.  
In \Cref{sec:Preliminary.theory} we summarize results concerning the  asymptotic performance of this trust-region algorithm. 

\subsection{A Noise Model}\label{sec:noisemodel}
Motivated by the role of the classical optimizer in the VQA setting, we consider a noise setting similar to that investigated in \citet{Cao2023}. 
In particular, we suppose that our access to an underlying objective function $f:\Reals^d\to\Reals$ is through what \citet{Cao2023} refers to as a ``stochastic zeroth-order oracle'' (of one of two types):

\begin{defi}
    \label{def:zoo}
    Let $\xi$ denote a random variable that may or may not depend on the optimization variables $\theta$. 
    We say $\tilde{f}(\theta, \xi)$ is a 
    \emph{zeroth-order oracle for $f$} 
    provided for all $x\in\Reals^d$, 
    the error quantity
    \begin{equation}
        \label{eq.error_quantity}
        e(\theta,\xi) = \tilde{f}(\theta,\xi) - f(\theta)
    \end{equation}
    satisfies at least one of two conditions: 
    \begin{description}
      \item[\normalfont{\textbf{Type 1. (Deterministically bounded noise)}}]\hlabel{type1}
        There exists a constant $\epsilon_f\geq 0$ such that $|e(\theta,\xi)|\leq\epsilon_f$ for all realizations of $\xi$. 

      \item[\normalfont{\textbf{Type 2. (Independent subexponential noise)}}]\hlabel{type2}
        There exist constants $\epsilon_f\geq 0$ and $c > 0$ such that 
        \begin{equation}
        \label{eq.subexp}
        \Prob_{\xi}\left[|e(\theta,\xi)|>t\right] \leq \exp(c(\epsilon_f - t)) \quad \forall t\geq 0.
        \end{equation}
    \end{description}
\end{defi}

We remark that the random variable $\xi$ in \Cref{def:zoo} should be considered exogenous in the sense that an optimization algorithm accessing $\tilde{f}$ cannot provide an input $\xi$. Hence, for ease of notation, we will denote the noisy evaluations of $f$ as $\tilde{f}(\theta)$ throughout this manuscript. 
We additionally remark that while, for generality, \Cref{def:zoo} is stated in terms of random variables, the definition does not preclude noisy deterministic functions;
in that case, one can trivially assign some Dirac distribution to $\xi$, and the definition is still meaningful. 
To provide context, one can interpret \cref{eq.subexp}  as saying that the errors exhibit a subexponential tail (with rate determined by $c$);
moreover, \cref{eq.subexp} suggests that there are no restrictions on the distribution of errors of magnitude within $\epsilon_f$. 
Many commonly used and naturally occurring probability distributions fall within this definition.

We believe  the assumption that a quantum computer's output exhibits the properties of a zeroth-order oracle is reasonable.
For example, in QAOA, a single shot of the output of the quantum circuit corresponds to one of combinatorially many binary vectors, which are then evaluated through the objective function of a combinatorial optimization problem on the classical device.  
The objective function used by the classical optimizer, in turn, is essentially an average of Bernoulli variables (one variable for each binary solution) weighted by the original combinatorial problem's objective  values. 
If one supposes that the combinatorial problem is well defined, and because Bernoulli variables have bounded (finite) support, the error $\epsilon_f$ is always trivially deterministically bounded in QAOA and is a \hyperref[type1]{Type 1} error as we have defined. 

Of course, such a trivial deterministic bound is loose to the point of uselessness.
However, using Bernstein inequality, one can show that the distribution of errors of the finitely supported $\tilde{f}$ will yield $c$ in \cref{eq.subexp} that scales linearly in the shot count (that is, the subexponential rate of decay is proportionally faster). 
\replace{Under this interpretation}{Given this observation}, \replace{}{we may replace}
$\epsilon_f$ \replace{}{with} the standard deviation $\sigma$ of the distribution 
described by $\tilde{f}(\theta, \xi)$, \replace{}{where $\xi$ describes the randomness associated with a given number of shots.} 

We are additionally attracted to the noise model defined in \Cref{def:zoo} in the VQA setting because of its nonparametric flexibility.
The definition itself does not specify what the noise distribution ought to be but essentially assumes\replace{ (via Bernstein inequality)}{, via conditions on the tail of the noise distribution,} only that the variance of $\tilde{f}(\theta)$ is defined and that there exist bounds on higher moments of $\tilde{f}(\theta)$. 

We also find this noise model attractive for VQA because of intermittent hardware noise, which remains a salient difficulty in near-term quantum devices. 
Owing to the nonparametric assumption, even if a quantum device ``drifts'' over time \cite{Proctor2020},
the noise model is sufficiently flexible to describe a sum of stochastic error (as previously discussed) and hardware error. 
While the ``variance'' (and higher moments) of hardware noise is not easily quantifiable, allowing a sufficiently large $\epsilon_f$ in the deterministically bounded regime of \Cref{def:zoo} can potentially encapsulate hardware noise.
This interpretation lends credence to the use and purported convergence guarantees of algorithms based on zeroth-order oracles, like the one we will describe now.  

\subsection{A Noise-Aware Trust-Region Algorithmic Framework}\label{sec:algorithm}

We begin by restating the general first-order trust-region algorithm of \citet{Cao2023}[Algorithm 1] in \Cref{alg.TR.noise}; our statement is identical up to changes in notation. 
For a complete statement of \Cref{alg.TR.noise}, we require one additional definition, which must be carefully handled in derivative-free optimization.
\begin{defi}
    \label{def:foo}
    Let $\xi$ denote a random variable.
    We say $g(\theta,\xi)$ is a \emph{first-order oracle} for $\nabla f(\theta)$ provided there exist $\kappa_{eg}\geq 0$ and $\epsilon_g\geq 0$ such that for any given $\theta\in\Reals^d$, for any given $\Delta>0$, and for any given probability $p_1\in[0.5,1]$,
    \begin{equation}
    \label{eq.gradient_bound}
    \Prob_{\xi}\left[\|g(\theta,\xi) - \nabla f(\theta)\|\leq\epsilon_g + \kappa_{eg}\Delta \right] \geq p_1 .
    \end{equation}
    The random variable $\xi$ may or may not depend on $\theta$ and $\Delta$. 
\end{defi}

Given access to a zeroth-order oracle (see \Cref{def:zoo}) and a first-order oracle (see \Cref{def:foo}), convergence results for \Cref{alg.TR.noise} may be proven.

\begin{algorithm2e}[H]
  \SetAlgoNlRelativeSize{-4}
    \KwIn{starting point $\theta_0$; initial trust-region radius $\Delta_0>0$; trust-region parameters
    $\eta_1, \eta_2, \gamma \in(0,1)$,
    tolerance parameter $r\geq 2\epsilon_f$, probability parameter $p_1$, bound on model Hessians $\kappa_{\text{bmh}}\geq 0$, constant $\kappa_{\mathrm{fcd}} \in (0,1]$ }
    Evaluate $\tilde{f}(\theta)$ if not previously evaluated.
    
    \For{$k=0,1,2, \ldots$}{    
       Compute model gradient $g_k$ via first-order oracle with input $\theta_k,p_1,$ and $\Delta_k$
       
       Let $H_k$ be a model Hessian satisfying $\|H_k\|\leq\kappa_{\text{bmh}}$
       
        Construct a local quadratic model 
        \begin{equation}
      \label{eq.model.def.quadratic}
    m_k(s) = \tilde{f}(\theta_k) + g_k^T s + s^T  H_k s
\end{equation}  

Compute $s_k$ as an approximate minimizer of $\{\min m_k(s): s\in\cB(0,\Delta_k) \}$, such that $s_k$ satisfies
\begin{equation}
  \label{eq.model.decrease.sufficient}
  m_k\left(\theta_k\right)-m_k\left(\theta_k+s_k\right) \geq \frac{\kappa_{\mathrm{fcd}}}{2}\left\|g_k\right\| \min \left\{\frac{\left\|g_k\right\|}{\left\|H_k\right\|}, \Delta_k\right\}
\end{equation}

Evaluate $\tilde{f}(\theta_k + s_k)$ from zeroth-order oracle as in  \Cref{def:zoo}.

Compute 
\begin{equation}
\label{eq.rho.k}
    \rho_k=\frac{\tilde f(\theta_k)-\tilde f(\theta_k+ s_k)+r}{m_k\left(0\right)-m_k\left(s_k\right)}
\end{equation}
        \eIf{$\rho_k \geq \eta_1$}{
            \begin{flalign*}
            \mbox{Set $\theta_{k+1}=\theta_k+s_k$ and }
                \Delta_{k+1}= 
                \begin{cases}\gamma^{-1} \Delta_k & \text { if }\left\|g_k\right\| \geq \eta_2 \Delta_k \\ 
                \gamma \Delta_k, & \text { if }\left\|g_k\right\|<\eta_2 \Delta_k
                \end{cases}&&
            \end{flalign*}
            }{
            Set $\theta_{k+1}=\theta_k$ and $\Delta_{k+1}=\gamma \Delta_k$.
        }
    }
\caption{General framework for noise-aware optimization\label{alg.TR.noise}}
\end{algorithm2e}

\subsection{Preliminary Assumptions and Analysis}\label{sec:Preliminary.theory}

\paragraphsummary{Present main theorem in \citet{Cao2023}}

Under reasonable assumptions, \citet{Cao2023}[Theorems 4.11 and 4.18] demonstrate that in both the deterministically bounded noise regime and the subexponential noise regime, the probability of exceeding $\mathcal{O}(1/\epsilon^2)$ iterations of \Cref{alg.TR.noise} to find an $\epsilon$-stationary solution to \cref{eq:basic_prob}  decays exponentially in the exceedance; $\epsilon$ is a function of $\epsilon_f$ and $\epsilon_g$.

We start by making the following assumptions. 
\begin{assumption}
     The objective function $f$ is continuously differentiable. That is, the gradient $\nabla f$ is $L_{\nabla f}$-Lipschitz continuous on $\mathcal{R}^d$ and satisfies
\begin{align*}
    \|\nabla f (\theta^{(1)})-\nabla f(\theta^{(2)})\|\le L_{\nabla f}\|\theta^{(1)}-\theta^{(2)}\| 
\end{align*}
for all $(\theta^{(1)}, \theta^{(2)}) \in \mathcal{R}^d \times \mathcal{R}^d$. 
\label{ass.L.smooth}
\end{assumption}

\begin{assumption} 
\label{ass.f.lower.bounded}
The function $f$ is lower bounded by $f_{\inf}$.
\end{assumption}

\begin{assumption} 
\label{ass.H.upper.bounded}
For all $k = 0,1,\cdots$, $\|H_k\| \le \kappa_{\mathrm{bhm}}$, where $\kappa_{\mathrm{bhm}} >0$. 
\end{assumption}

We state intentionally simplified versions of two main results of \citet{Cao2023}, which provide the number of iterations $T$ needed to ensure (with high probability) that within the first $T$ iterations, the event  $\left \{ \left\|\nabla f\left(\theta_k\right)\right\| \leq \epsilon \right\}$ occurs.

\begin{theorem}
\label{theorem.complexity.cao.1}
Suppose \Crefrange{ass.L.smooth}{ass.H.upper.bounded} are satisfied. Suppose we have access to a 
zeroth-order oracle of \hyperref[type1]{Type 1} with parameter $\epsilon_f$ and a
first-order oracle with parameters $\kappa_{eg}, \epsilon_g,$ and $p_1$. 
Let $\{\Theta_k\}$ denote the sequence of random variables with realizations $\{\theta_k\}$ generated by \Cref{alg.TR.noise}.
There exists $\kappa$, independent of $\epsilon_f, \epsilon_g, p_1$ but dependent on $\kappa_{eg}$, such that
given any
$$\epsilon > \kappa\left(\sqrt{\frac{\epsilon_f}{2p_1 - 1}}
+ \epsilon_g \right),$$
it holds that
$$\Prob\left[\displaystyle\min_{0\leq k \leq T-1} \|\nabla f(\Theta_k)\| \leq \epsilon\right] 
\geq
1 - \exp\left(-\mathcal{O}(T)\right)$$
for any $T\in \mathcal{O}\left(\frac{1}{\epsilon^2}\right)$.

\end{theorem}

\begin{theorem}
\label{theorem.complexity.cao.2.simple}
Suppose that  \Crefrange{ass.L.smooth}{ass.H.upper.bounded} are satisfied. Suppose we have access to a 
zeroth-order oracle of \hyperref[type2]{Type 2} with parameters $\epsilon_f$ and $c$ and a first-order oracle with parameters $\kappa_{eg}, \epsilon_g,$ and $p_1$. Define $p_0 = 1- 2 \exp(c (\epsilon_f -r/2 ))$. 
Let $\{\Theta_k\}$ denote the sequence of random variables with realizations $\{\theta_k\}$ generated by \Cref{alg.TR.noise}. There exists $\kappa$, independent of $\epsilon_f, \epsilon_g, p_1$ but dependent on $\kappa_{eg}$, such that
given any
$$\epsilon > \kappa\left(\sqrt{\frac{\epsilon_f + 1/c}{2 p_0+2 p_1-3}}
+ \epsilon_g \right),$$
it holds that
$$\Prob\left[\displaystyle\min_{0\leq k \leq T-1} \|\nabla f(\Theta_k)\| \leq \epsilon\right] 
\geq
1 - 2\exp\left(-\mathcal{O}(T)\right)- \exp\left(-ct/4\right) $$
for any $T\in \mathcal{O}\left(\frac{1}{\epsilon^2}\right)$ and any $t > 0$.
\end{theorem}

\begin{remark}
While we encourage the reader to study \citet{Cao2023}[Theorems 4.11 \& 4.18] for a complete description of the constants that we have deliberately hidden in $\mathcal{O}(T)$ and $\mathcal{O}(1/\epsilon^2)$, 
we remark that, as one might intuit, the $\mathcal{O}(T)$ rate in the exponentially decaying probability decreases in $\epsilon_f$, 
and the $\mathcal{O}\left(\frac{1}{\epsilon^2}\right)$ term increases with both the initial optimality gap $f(\theta_0) - f_{\inf}$ and the Lipschitz constant $L_{\nabla f}$. When $\epsilon_f = \epsilon_g = 0$, \Cref{theorem.complexity.cao.1} reduces to something resembling the first-order convergence result for a standard derivative-free trust-region algorithm; see, for example, \citet{conn2009introduction}[Chapter 10]. \replace{}{For theoretical soundness, we set $r$ to be any value larger than $2 \epsilon_f$. When $r < 2 \epsilon_f$, $\rho_k$ could be dominated by the noise and the algorithm may fail to make progress. Since the lower bounds on $\epsilon$ in \Cref{theorem.complexity.cao.1} and \Cref{theorem.complexity.cao.2.simple} depend on $r$, $r =  2 \epsilon_f$ is the theoretically optimal choice of $r$, provided $\epsilon_f$ was exactly knowable.}
\end{remark}

\section{A Practical Noise-Aware Derivative-Free Algorithm Based on \Cref{alg.TR.noise}}\label{sec:novelalg}

While \Cref{alg.TR.noise}
provides an invaluable framework for analyzing DFO methods, 
it does not lend itself immediately to a practical algorithm for DFO. 
In particular, \Cref{alg.TR.noise} does not detail how to construct a first-order oracle for computing $g_k$ or how to choose model Hessians $H_k$. 
\citet{Cao2023} suggest using linear interpolation of noisy function values for the first-order oracle, similar to the approach analyzed by \citet{berahas2022theoretical}. 
Our approach aligns with this line of reasoning but is further motivated by two practical considerations.
First, we aim to go beyond linear interpolation models to also minimum Frobenius norm quadratic interpolation models.
(This extension is certainly not new and is the subject of \citet{conn2009introduction}[Chapter 5].) \replace{}{Compared to linear interpolation models, minimum Frobenius norm models are more complex and generally model the true objective with greater accuracy when the objective is nonlinear; this typically leads to more per-iteration objective decrease. Additionally, minimum Frobenius norm models do not require as many samples as quadratic interpolation models to construct. 
}
Second, and more important, because we often employ model-based methods for optimizing computationally expensive objectives in derivative-free optimization, we do not want to generate a new set of interpolation points on every iteration. Rather, we prefer to judiciously reuse previously evaluated points and their corresponding (noisy) function evaluations. 
Especially in the noisy setting, this merits a careful reexamination of known results concerning the geometry of interpolation sets and algorithmic methods for ensuring geometric constraints are satisfied in practice. In addition, we impose a lower bound on a sampling radius when we select points for model construction, in order to ensure that the error between the true gradient and the model gradient is controlled.  
\Cref{sec:Geometry} focuses on precisely establishing these algorithmic controls. 

\subsection{Geometry of Interpolation Sets}
\label{sec:Geometry}
At every iteration $k \in \{0,1,\dots\}$ of our derivative-free trust-region algorithm, there exists an incumbent iterate $\theta_k$ (that is also a trust-region center).
In each iteration, a set of distinct points $\cX = \{x^0 = \theta_k, x^1, \cdots, x^p\}\subset\Reals^d$ is selected; the algorithm ensures that the noisy function $\tilde{f}$ is evaluated at each $x\in\cX$. 
A model is constructed such that it interpolates the noisy function $\tilde{f}$ at each $x\in\mathcal{X}$, and this model is intended to approximate the objective function near $\theta_k$.

In this work we will focus on minimum Frobenius norm quadratic interpolation models, which are appropriate in the setting of budget-constrained (expensive) derivative-free optimization.  
We consider the space of all quadratic polynomials in $\Reals^d$;
a suitable basis for this space is given by the union of the degree-zero and degree-one monomials $\{1, x_1, x_2, \dots, x_d\}$ and the degree-two monomials $\{x_1^2,\dots,x_d^2,x_1x_2,x_1x_3,\dots,x_1x_n,\dots, x_{n-1}x_n\}$. 
Vectorizing these two sets as $\mu(x)\in\Reals^{d+1}$ and $\nu(x)\in\Reals^{d(d+1)/2}$ respectively, 
we see that any degree-two polynomial, and hence any quadratic interpolation model, can be expressed as
 \begin{align}
 \label{eq.mfn_model}
m(x)=\alpha^T \mu(x)+\beta^T \nu(x),
\end{align} 
for $\alpha\in\Reals^{d+1}$ and $\beta\in\Reals^{d(d+1)/2}$. 

We assume throughout that the interpolation set $\cX$ has cardinality $p \in [d+1,(d+1)(d+2)/2]$. 
Define 
\begin{equation}
    \tilde f(\mathcal{X})=\left[\begin{array}{r}
\tilde f\left(x^0\right),
\tilde f\left(x^1\right),
\dots,
\tilde f\left(x^p\right)
\end{array}\right]^\top,
\end{equation}
where $\tilde f\left(x^i\right)$ represents noisy evaluations of $ f\left(x^i\right)$ for $i = 0, 1,\dots, p$. 
Enforcing $m(x^i) = \tilde  f(x^i)$ for $i=0,1,\dots,p$ is tantamount to solving
 \begin{align}
 \label{eq.model.interpolation}
\left[\begin{array}{l}
M_{\mathcal{X}} \\
N_{\mathcal{X}}
\end{array}\right]^T\left[\begin{array}{l}
\alpha \\
\beta
\end{array}\right]=\tilde  f(\mathcal{X}),
\end{align} 
where we define $M_{\mathcal{X}} \in \mathcal{R}^{(d+1) \times|\mathcal{X}|}$ and $N_{\mathcal{X}} \in \mathcal{R}^{d(d+1) / 2 \times|\mathcal{X}|}$, by $M_{i, j}=\mu_i\left(x^j\right)$ and $N_{i, j}=\nu_i\left(x^j\right)$, respectively. 
Note that \cref{eq.model.interpolation} may not admit a unique solution since $| \mathcal{X}| \le (d+1)(d+2)/2$. 
We will focus on solutions to the interpolation problem \cref{eq.model.interpolation} that are of minimum norm with respect to the vector $\beta$. 
That is,
we seek $(\alpha, \beta)$ that solve
 \begin{align}
 \label{eq.MNH}
\min \left\{\frac{1}{2}\|\beta\|^2: M_{\mathcal{X}}^T \alpha+N_{\mathcal{X}}^T \beta=\tilde f(\mathcal{X}) \right\}.
\end{align} 

The solution to \cref{eq.MNH} is a quadratic polynomial whose Hessian matrix is of minimum Frobenius norm, since $\|\beta\|=\left\|\nabla_{xx}^2 m(x)\right\|_F$. 
The KKT conditions for \cref{eq.MNH} can be written as
 \begin{align}
 \label{eq.KKT}
\left[\begin{array}{cc}
N_{\mathcal{X}}^T N_{\mathcal{X}} & M_{\mathcal{X}}^T \\
M_{\mathcal{X}} & 0
\end{array}\right]\left[\begin{array}{l}
\lambda \\
\alpha
\end{array}\right]=\left[\begin{array}{c}
\tilde f(\mathcal{X})  \\
0
\end{array}\right]
\end{align} 
with $\beta = N_{\mathcal{X}} \lambda$. 
We make the observation that to guarantee \cref{eq.KKT} will yield a unique solution, we must have that both $rank(M_{\mathcal{X}}) = d + 1$
and 
$N_{\mathcal{X}}^T N_{\mathcal{X}}$ is positive definite for all $u \in \mathcal{R}^d$ such that $M_{\mathcal{X}} u = 0$. 

Although relatively expensive compared with, for instance, the heuristic employed in \texttt{POUNDers} (see \citet{wild2008mnh} for details), we find that careful geometry maintenance is critical to the performance of a derivative-free model-based optimization method in the noisy setting.
Recalling \Cref{def:zoo}, we aim to ensure that, given $\Delta>0$, the model gradient $g$ defined by $\alpha$ in \cref{eq.mfn_model} satisfies \cref{eq.gradient_bound}
for some values of $\epsilon_g, \kappa_{eg},$ and $p_1$. 
Demonstrating such a result requires careful consideration of the geometry of $\cX$. We begin by introducing the concept of Lagrange polynomials associated with $\cX$. 

\begin{defi}
Given a set $\cX = \{x^0, x^1, \dots, x^p\} \subset\Reals^d$ of interpolation points, we define the \emph{minimum Frobenius norm Lagrange polynomials associated with $\cX$} as the set of polynomials 
$\{\ell_i(x) = \alpha_i^\top \mu(x) + \beta_i^\top \nu(x): i=0,1,\dots,p\}$
with coefficients $(\alpha_i,\beta_i)$ chosen such that, for each $i$,
$(\alpha_i,\beta_i)$ is the solution to 
\cref{eq.MNH} when the function $f$ in the right-hand side in the constraint is replaced with the indicator function for $x^i$. 
\end{defi}

It is straightforward to demonstrate that the minimum Frobenius norm model $m(x)$ in \cref{eq.mfn_model} is expressible via a linear combination of minimum Frobenius norm Lagrange polynomials, that is,
\begin{equation}
    \label{eq.mfn_poly}
    m(x) = \displaystyle\sum_{i=0}^p f(x^i)\ell_i(x).
\end{equation}
Moreover,
denoting the matrix in the left-hand side of \cref{eq.KKT} by $W_{\cX}$, it is clear from the definition of minimum Frobenius norm Lagrange polynomials that each $(\alpha_i, \beta_i)$ defining the polynomials may be computed directly from the columns of $W_{\cX}^{-1}$. 
In particular, the last $d+1$ entries of the $i$th column of $W_{\cX}^{-1}$ define $\alpha_i$.
Also, by considering the constraints in the quadratic dual program to \cref{eq.MNH}, 
one may derive that the quadratic term of $\ell_i(x)$ described by $\beta_i$ is a weighted sum of rank-one matrices, 
$$\nabla^2\ell_j(x)
=
\displaystyle\sum_{j=0}^p \left[W_{\cX}^{-1}\right]_{j + 1, i} (x^i - x^0)(x^i - x^0)^\top. 
$$

With computable expressions for the minimum Frobenius norm Lagrange polynomials associated with a set $\cX$ in hand, the following definition of $\Lambda$-poisedness is easily stated.
\begin{defi}
We say that $\cX$ is \emph{poised in the minimum Frobenius norm sense} provided $W_{\cX}$ is nonsingular. 
\end{defi}

\begin{defi}
\label{def.lambdapoised}
Let $\Lambda>0$, and let a set $B\subset\Reals^d$ be given. 
A set $\cX = \{x^0,x^1,\dots,x^p\}$, poised in the minimum Frobenius norm sense, 
is moreover \emph{$\Lambda$-poised in $B$ in the minimum Frobenius norm sense} provided

$$\Lambda \geq \displaystyle\max_{0,1,\dots,p} \displaystyle\max_{x\in B} |\ell_i(x)|.$$
\end{defi}

Critical to our development of error bounds, 
we state the following results.

\begin{theorem}
\label{thm.mnh_error}
Let $\Lambda>0$, and let 
$\cX=\{x^0,x^1,\dots,x^p\}$ be a given set of interpolation points. 
Let
$\cB(x^0,\Delta) \subset\Reals^d$ be a ball sufficiently large such that $\cX\subset \cB(x^0,\Delta)$. 
Suppose \Cref{ass.L.smooth} holds with constant $L_{\nabla f}$. 
Define the matrix
$$\hat{L} = \frac{1}{\Delta}
\left[
\begin{array}{cccc}
x^1 - x^0 & x^2 - x^0 & \cdots & x^p - x^0 ,\\
\end{array}
\right]^\top
$$
and denote its pseudoinverse via $\hat{L}^\dagger = (\hat{L}^\top \hat{L})^{-1}\hat{L}^\top$. 
Define 
$$\epsilon_{\max} = \displaystyle\max_{i=1,2,\dots,p} |\tilde{f}(x^i) - f(x^i)| 
\quad \text{ and } \quad
\epsilon_0 = |\tilde{f}(x^0) - f(x^0)|. 
$$

If $\cX$ is $\Lambda$-poised in $\cB(x^0,\Delta)$ in the minimum Frobenius norm sense, 
then the minimum Frobenius norm model $m(x)$ satisfies
$$\|\nabla f(x) - \nabla m(x)\| \leq \sqrt{p + 1} \|\hat{L}^\dagger\|
\left[
(L_{\nabla f} + \|\nabla^2 m(x^0)\|)\Delta 
+ \frac{\epsilon_0 + \epsilon_{\max}}{\Delta}
\right]
\quad \forall x\in \cB(x^0,\Delta).$$
\end{theorem}

\begin{proof}{\textbf{Proof.}} 
    For ease of notation, we denote the quadratic underdetermined interpolation model $m(x)$ as
    $$m(x) = c + g^\top x + \frac{1}{2}x^\top H x.$$
    In other words, we have simply reorganized the terms in \cref{eq.mfn_model}. 
    For an arbitrary point $x\in \cB(x^0,\Delta)$, denote the function value error ($e^f$) and the gradient value error ($e^g$) via
    $$m(x) = f(x) + e^f(x) \quad \text{ and }
    \quad \nabla m(x) = Hx + g = \nabla f(x) + e^g(x). $$
    Because $m(x)$ interpolates $\tilde{f}(x)$ at each point in $\cX$,
    we conclude from the expression for $e^f$ that, for each $i=0,1,\dots,p$, 
    $$
    \begin{array}{rll}
    m(x^i) & = \tilde{f}(x^i) & \iff\\
    c + g^\top x^i + \frac{1}{2}x^{i\top} H x^i & = \tilde{f}(x^i) & \iff\\
        c + g^\top x^i + \frac{1}{2}x^{i\top} H x^i - m(x) & = \tilde{f}(x^i) - m(x)& \iff\\
    (x^i - x)^\top g + \frac{1}{2}(x^i - x)^\top H(x^i- x) + (x^i - x)^\top H x & = \tilde{f}(x^i) - f(x) - e^f(x).
    \end{array}
    $$
    Now, substituting in the expression for $e^g$, we have for each $i=0,1,\dots,p$
    $$
    (x^i - x)^\top (e^g(x) + \nabla f(x)) + \frac{1}{2}(x^i-x)^\top H(x^i-x) 
    = \tilde{f}(x^i) - f(x) - e^f(x).
    $$
    Trivially, we also have
    $$
    (x^i - x)^\top (e^g(x) + \nabla f(x)) + \frac{1}{2}(x^i-x)^\top H(x^i-x) 
    = \tilde{f}(x^i) - f(x) - e^f(x) + f(x^i) - f(x^i).
    $$
    
    By Taylor's theorem and the definition of $\Delta$, $f(x) + \nabla f(x)^\top(x^i - x) - f(x^i) \in \mathcal{O}(\Delta^2)$. Thus there exists $c_i\in\Reals$ independent of $\Delta$ such that, for all $i=0,1,\dots,p$,
    \begin{equation}\label{eq.eq1}
    (x^i-x)^\top e^g(x) + \frac{1}{2}(x^i-x)^\top H(x^i-x) = c_i\Delta^2  - e^f(x)+ \tilde{f}(x^i) - f(x^i).
    \end{equation}
    Now, subtracting \cref{eq.eq1} for $i=0$ from each of \cref{eq.eq1} for $i=1,\dots,p$, we get for $i=1,\dots,p$ that
    $$
    (x^i-x^0)^\top e^g(x) + \frac{1}{2}(x^i - x)^\top H (x^i - x) - \frac{1}{2}(x^0 - x)^\top H (x^0 - x) = 
    (c_i - c_0) \Delta^2 + \tilde{f}(x^i) - f(x^i) - \tilde{f}(x^0) + f(x^0).
    $$
    It now follows from the assumptions on $f$ and $\tilde{f}$ and the definitions of $\Delta, \hat{L}^\dagger, \epsilon_{\max}$ and $\epsilon_0$ that 
    $$\|e^g(x)\| \leq \sqrt{p+1}\|\hat{L}^\dagger\| 
    \left[
    \left(L_{\nabla f} + \|H\|\right)\Delta + 
    \left(\frac{\epsilon_0 + \epsilon_{\max}}{\Delta}\right)
    \right]. 
    $$
    \qed
\end{proof}

\replace{}{In} order to extract a meaningful bound from \Cref{thm.mnh_error}, it remains to bound both the quantities $\|\hat{L}^\dagger\|$ and 
 $\|\nabla^2 m(x^0)\|$ appearing in \Cref{thm.mnh_error}.
 The former quantity can be bounded by observations made in the development of \citet{conn2009introduction}[Section 5.3].
The proof of these observations involves an alternative characterization of \Cref{def.lambdapoised} that we will not provide in this manuscript for the sake of concise exposition; see \citet{conn2009introduction}[Section 5.3] for full details. 
We state this bound on $\|\hat{L}^\dagger\|$ in \Cref{prop.bound_pinv}. 
The quantity $\|\nabla^2 m(x^0)\|$ can be bounded by mimicking the proof of a result in
\citet{conn2009introduction}[Theorem 5.7]. 
We state and prove that bound in \Cref{thm.model_hessian_bound}. 

\begin{proposition}
\label{prop.bound_pinv}
    Let $\cX$ be $\Lambda$-poised in the minimum Frobenius norm sense. 
    Let \Cref{ass.L.smooth} hold. 
    Then, $\hat{L}^\dagger$ in the statement of \Cref{thm.mnh_error} automatically satisfies
    $$\| \hat{L}^\dagger\| \leq \sqrt{p+1}\Lambda.$$ 
\end{proposition}

\begin{theorem}
\label{thm.model_hessian_bound}
 Let $\Lambda>0$, and let 
$\cX=\{x^0,x^1,\dots,x^p\}$ be a given set of interpolation points. 
Denote the absolute error in function value at each interpolation point by 
$\epsilon_i = |e(x^i,\xi_i)|.$
Let
$\cB(x^0,\Delta) \subset\Reals^d$ be a ball sufficiently large such that $\cX\subset \cB(x^0,\Delta)$. 
Suppose $f$ is continuously differentiable in an open set $\Omega$ such that $\cB(x^0,\Delta)\subset\Omega$, and suppose $\nabla f$ is Lipschitz continuous in $\Omega$ with constant $L_{\nabla f}$.
Then,
\begin{equation}\label{eq.kappa_bmh}\|\nabla^2 m(x^0)\| \leq 
\displaystyle\frac{4(p+1)\Lambda L_{\nabla f}\sqrt{(d+1)(d+2)}}{\nu(\Delta)}
+ \displaystyle\frac{8\Lambda\sqrt{(d+1)(d+2)}}{\Delta^2\nu(\Delta)}\sum_{i=0}^p \epsilon_i,
\end{equation}
where we have denoted 
\begin{equation}\label{def.nu}\nu(\Delta) = \min\left\{1,\frac{1}{\Delta},\frac{1}{\Delta^2}\right\}.
\end{equation}
\end{theorem}

\begin{proof}{\textbf{Proof.}} Much of the proof of \citet{conn2009introduction}[Theorem 5.7] holds here. 
    In particular, 
    $$\|\nabla^2 \ell_i(x)\| \leq \displaystyle\frac{8\Lambda\sqrt{(d+1)(d+2)}}{\Delta^2\nu(\Delta)}$$
    holds for each $\ell_i(x)$, because it is only a property of minimum Frobenius norm Lagrange polynomials and is independent of any assumptions on noise. 

    As in the proof of \citet{conn2009introduction}[Theorem 5.7], we can assume without loss of generality that $f(x^0) = 0$ and that $\nabla f(x^0) = 0$ by subtracting an appropriate linear polynomial from $f$. 
    If we subtract the same linear polynomial from the minimum Frobenius norm model $m(x)$, then the model Hessian $\|\nabla^2 m(x)\|$ remains unchanged. 

    By our subtraction of the linear polynomial, we have from Taylor's theorem that
    $$\displaystyle\max_{x\in \cB(x^0,\Delta)} |f(x)| \leq \frac{L_{\nabla f}}{2}\Delta^2. $$
    Thus, by \cref{eq.mfn_poly},
    $$
    \begin{array}{rl}
    \|\nabla^2 m(x^0)\| \leq \displaystyle\sum_{i=0}^p |\tilde{f}(x^i)| \|\nabla^2 \ell_i(x)\|
    \leq & \displaystyle\sum_{i=0}^p (|f(x^i)| + \epsilon_i) \|\nabla^2 \ell_i(x)\| \\
    \leq & \displaystyle\sum_{i=0}^p 
    \left[\frac{4\Lambda L_{\nabla f}\sqrt{(d+1)(d+2)}}{\nu(\Delta)} + 
    \frac{8\Lambda\sqrt{(d+1)(d+2)}\epsilon_i}{\Delta^2\nu(\Delta)}
    \right],\\
    \end{array}
    $$
    as we meant to show. 
    \qed
\end{proof}

\Cref{thm.model_hessian_bound} establishes that the norm of the model Hessian is upper bounded, either deterministically or with high probability, when $\tilde f$ is a zeroth-order oracle as defined in \Cref{def:zoo}.
\Cref{thm.mnh_error}, \Cref{prop.bound_pinv}, and \Cref{thm.model_hessian_bound} taken together motivate two algorithmic features that we now discuss. 

\subsubsection{Decoupling the Sampling Radius from the Trust Region Radius}
\label{sec:decoupling}
Combining the results in \Cref{thm.mnh_error}, \Cref{prop.bound_pinv}, and \Cref{thm.model_hessian_bound},
we see that a bound on $\|\nabla f(x) - \nabla m(x)\|$ is minimized (as a function of $\Delta$) provided
$$\Delta = \min\left\{\displaystyle\sqrt{
\frac{\epsilon_0 + \epsilon_{\max} + 8\Lambda\sqrt{(d+1)(d+2)}\sum_{i=0}^p\epsilon_i}
{L_{\nabla f}(1 + 4(p+1)\sqrt{(d+1)(d+2)})}
},1\right\}.$$ 
If we make the coarse assumption that 
each $\epsilon$ term $(\epsilon_i, \epsilon_0, \epsilon_{\max})$ is bounded by $\mu\epsilon_f$ for some small multiple $\mu\geq 1$ 
(for example, this is true with $\mu=1$ by definition in the deterministically bounded noise regime),
then we can simplify this result to 
$$\Delta \approx 
\min\left\{\sqrt{\frac{2 + 8\Lambda(p+1)\sqrt{(d+1)(d+2)}}
{1 + 4(p+1)\sqrt{(d+1)(d+2)}}}
\sqrt{\frac{\mu\epsilon_f}
{L_{\nabla f}}
},1\right\}.$$
Moreover, assuming $\Lambda\approx 1$ (which we discuss more in the next subsection), this greatly simplifies to
$$\Delta \approx \displaystyle\sqrt{
\frac{2\mu\epsilon_f}
{L_{\nabla f}},
}$$
which is generally less than 1, since we expect $\epsilon_f$ to be small relative to $L_{\nabla f}$. 
Motivated by this observation, our algorithm explicitly decouples the trust-region radius from the sampling radius.
Our algorithm maintains running estimates, $\tilde\epsilon_f$ and $\tilde L_{\nabla f}$, of $\epsilon_f$ and $L_{\nabla f}$, respectively.
Given parameters $\mu\geq 1$ and $\bar\Lambda\geq 1$, if the current trust region is $\cB(x^k,\Delta_k)$, then our algorithm will ensure that a set of interpolation points $\cX$ is $\bar\Lambda$-poised on the set $\cB(x^k, \max\{\Delta_k, \sqrt{2\mu\tilde\epsilon_f/\tilde{L}_{\nabla f}}\})$. 
This is a departure from standard model-based methods. 

\subsubsection{Careful Maintenance of Poisedness}
Combining the results in \Cref{thm.mnh_error}, \Cref{prop.bound_pinv}, and \Cref{thm.model_hessian_bound}, we see that 
the first-order error made by a minimum Frobenius norm model on a ball of radius $\Delta$, with a set of interpolation points that are $\Lambda$-poised on that ball, scales quadratically with $\Lambda$ in the worst case. 
Thus, when selecting a set $\cX$ of interpolation points, we intend to keep $\Lambda$ bounded.
This is achievable by employing an analog of \citet{conn2009introduction}[Algorithm 6.3], which we state in \Cref{alg.improve_poisedness}.

\begin{algorithm2e}
  \SetAlgoNlRelativeSize{-4}
    \KwIn{Desired upper bound on poisedness $\bar\Lambda > 1$, initial  set of points $\cX = \{x^0, x^1, \dots, x^p\}$ poised in the minimum Frobenius norm sense, radius $\Delta$.}
    
    \For{$k=1,2, \ldots$}{    
         Obtain minimum Frobenius norm Lagrange polynomials for $\cX$, $\{\ell_i(x): i=0,1,\dots,p\}$. \label{line.mfnlp}
           
       Compute \label{line.trsps}
       $$\Lambda_{k-1} = \displaystyle\max_{i=0,1,\dots,p} \max_{x\in \cB(x^0,\Delta)} |\ell_i(x)|.$$

       \eIf{$\Lambda_{k-1} > \bar\Lambda$ \label{line.arewedone}}{
        $i_k \gets\arg\displaystyle\max_{i=0,1,\dots,p} \max_{x\in \cB(x^0,\Delta)} |\ell_i(x)|$

        $x^+\gets \arg\displaystyle\max_{x\in \cB(x^0,\Delta)} |\ell_{i_k}(x)|$

        $\cX\gets \cX\cup\{x^+\}\setminus\{x^{i_k}\}$
       }
       {
        $\textbf{Return } \cX$, a $\bar\Lambda$-poised set. 
        }
    }
\caption{Improving poisedness of $\cX$ \label{alg.improve_poisedness}}
\end{algorithm2e}

It is proven in \citet{conn2009introduction}[Theorem 6.6] that \Cref{alg.improve_poisedness} terminates finitely, provided one solves the trust-region subproblems in \Cref{line.trsps} with sufficient accuracy  to be able to properly evaluate the condition in \Cref{line.arewedone}.
We also highlight that \Cref{alg.improve_poisedness} requires the initial $\cX$ to be poised in the minimum Frobenius norm sense. 
This condition is easily tested; 
if $\cX$ is not poised in the minimum Frobenius norm sense, then the inversion of $W_{\cX}$ in \eqref{eq.KKT} required in \Cref{line.mfnlp} will fail. 
In the rare event that $W_{\cX}$ becomes singular to working precision, our optimization method will employ an existing routine found in \texttt{POUNDers} \cite{wild2017chapter} to select a set of $d$ points from the history of all previously evaluated points such that 
(1) the set of $d$ points taken together with $x^0$ form an affinely independent set and
(2) $\|d - x^0\|$ is bounded by a small multiple of $\Delta$. 
We then replace $\cX$ with this set of $d+1$ many points. 
This algorithm is provided in \Cref{alg.affpoints}. 

We remark that \citet{Powell2003} is entirely dedicated to suggestions for developing a method similar to \Cref{alg.improve_poisedness} that is far more efficient, 
in particular avoiding the solution to $\mathcal{O}(p)$ many trust-region subproblems in \Cref{line.trsps}
and 
moreover avoiding explicit storage and inversions of $W_{\cX}$ in \Cref{line.mfnlp} by employing low-rank updates. 
Implementing Powell's suggestions is an avenue for future code development that we intend to pursue.
However, to be absolutely consistent with the theory in this paper, we effectively produced a direct implementation of \Cref{alg.improve_poisedness}. 
This renders the per-iteration linear algebra cost of our method very high compared with standard implementations of model-based optimization methods.
However, the metric we are concerned with in this paper is limiting the number of oracle calls, not per-iteration linear algebra costs, appropriate for real-world situations where oracle calls are computationally expensive. 

\subsection{Statement of DFO Algorithm}
We now explicitly provide an algorithm in the framework of \Cref{alg.TR.noise} appropriate for expensive derivative-free optimization. 
This algorithm is provided in \Cref{alg.DFOTR.noise}.

\Cref{alg.DFOTR.noise} uses $\cX$, free of any subscript notation, to denote an interpolation set; this is intended to remind the reader that we will reuse previously evaluated points as much as reasonably possible, and so $\cX$ generally transcends the iteration count. 
In every iteration $k \in \{0,1,\dots\}$, we first obtain a noise estimate $\tilde\epsilon_f$ and a local gradient Lipschitz constant estimate $\tilde{L}_{\nabla f}$;
we leave the means to compute these estimates intentionally vague in \Cref{alg.DFOTR.noise}, but in \Cref{sec:results} we discuss how one can do this in practice. 
While \Cref{alg.DFOTR.noise} employs the same trust-region mechanism with trust-region radius $\Delta_k$ as in \Cref{alg.TR.noise}, \Cref{alg.DFOTR.noise} explicitly decouples $\Delta_k$ from a \emph{sampling radius} $\bar\Delta_k \geq 
\sqrt{r\tilde\epsilon_f/\tilde{L}_{\nabla f}}$ for some tolerance parameter $r\geq 2$.
This decoupling is done in order
to optimize the accuracy of the model gradient, a principle elucidated in \Cref{sec:decoupling}.   
We then update $\cX$ by removing from $\cX$ any point $x$ satisfying $\|x - \theta_k\| > c_s\bar\Delta_k$, for some parameter $c_s > 1$, to yield the interpolation set $\cX = \{x^0 = \theta_k, x^1, \cdots, x^p \}$. 
As described in \Cref{sec:Geometry}, the cardinality of the interpolation set must satisfy $|\cX| \in [d+1,(d+1)(d+2)/2]$. When $|\cX| > (d+1)(d+2)/2$, we must remove superfluous interpolation points; we elect to remove the  $| \cX| - (d+1)(d+2)/2$ oldest points---oldest referring to the history of when points were added to $\cX$---from $ \cX$.
Let $\cX - \theta_k := \begin{bmatrix}
x^1 - \theta_k & \dots & x^q - \theta_k
\end{bmatrix}$, the matrix of displacements of each interpolation point from $\theta_k$, omitting the zero column resulting from $x^0=\theta_k$.
If  $\cX - \theta_k$ is not full-rank, then we augment  $\cX$ to contain an affinely independent set of points using \Cref{alg.affpoints}. 
We then call \Cref{alg.improve_poisedness} to ensure that  $\cX$ is $\bar\Lambda$-poised in $\cB(\theta_k, \bar\Delta_k)$ in the minimum Frobenius norm sense.

After constructing $\cX$, $\tilde f(x)$ is evaluated 
at any $x \in \cX$ at which $\tilde f$ was not previously evaluated. 
Then, we solve \cref{eq.MNH} (via \cref{eq.KKT}) and use the obtained $\alpha, \beta$ to construct a local quadratic model 
       \begin{equation}
      \label{eq.model.def.quadratic.alg}
    m_k(s) = \tilde{f}(\theta_k) + g_k^T s + s^T  H_k s.
     \end{equation} 
The trust-region subproblem $\{\min m_k(s) : s\in\cB(0, \Delta_k)\} $
is then solved accurately enough
 that the step $s_k$ satisfies the same Cauchy decrease condition  \cref{eq.model.decrease.sufficient} as in \Cref{alg.TR.noise}. 
We remark that attaining \cref{eq.model.decrease.sufficient} is trivial in practice,  since when $\kappa_{\mathrm{fcd}} = 1$, the Cauchy step will suffice.
We then evaluate $\tilde{f}$ at the test point $\theta_k+s_k$ and augment $\cX$ with the test point. 

Next, we compute $\rho_k$ defined as in \cref{eq.rho2.k}. 
We highlight that \cref{eq.rho2.k} (and \cref{eq.rho.k} in \Cref{alg.TR.noise}) is different from the ratio of actual decrease to predicted decrease employed in
classical trust-region methods;
in particular there is an additional factor of $\epsilon_f$ in the numerator of \cref{eq.rho2.k}.
The \emph{relaxed ratio} \cref{eq.rho2.k} was proposed to account for the noise in $\tilde f(\theta_k)$ and $\tilde f(\theta_k+ s_k)$ so that the numerator is not dominated by noise when $\Delta_k$ is small. 
A similar strategy was also considered by \citet{sun2022trust}, except they chose to add an $\epsilon$ term to both the numerator and denominator. 

The criteria to determine whether the test point is accepted or not and the rule of updating $\Delta_k$ are presented
between \cref{line:trust_region_adjustment} and \cref{line:trust_region_adjustment2}. 
It is unchanged from \Cref{alg.TR.noise} in \citet{Cao2023}.

\begin{algorithm2e}[H]
    \KwIn{starting point $\theta_0\in\Reals^d$; initial trust-region radius $\Delta_0>0$; upper bound on trust-region radius $\Delta_{\max}\geq \Delta_0$; trust-region parameters,
    $\eta_1, \eta_2, \gamma \in(0,1)$; tolerance parameter $r\geq 2$;
    initial interpolation set $\cX \supseteq \{x^0 =\theta_0\}$; sampling constant $c_s > 1$; maximum poisedness constant $\bar\Lambda$. }

    \For{$k=0,1,2, \ldots$}{ 
           
        \label{line:get_noise_estimate} Obtain noise estimate $\tilde\epsilon_f$ and local gradient Lipschitz constant estimate $\tilde{L}_{\nabla f}$. 
        Ensure $\tilde{L}_{\nabla f}$ chosen large enough so that $\tilde{L}_{\nabla f} \geq r\tilde\epsilon_f$.\\ \label{line:first_step}
       
        $\bar\Delta_k \gets \max\left\{\Delta_k, \displaystyle\sqrt{\frac{r\tilde\epsilon_f}{\tilde{L}_{\nabla f}}}
        \right\}$.\\
        
           Remove from $\cX$ any point $x$ satisfying
       $\|x - \theta_k\| > c_s\bar\Delta_k$.

       \If{$|\cX| > (d+1)(d+2)/2$}{
       Remove the $|\cX| - (d+1)(d+2)/2$ oldest obtained points from $\cX$. 
       }

       \If{$\rank(\cX - \theta_k) < d$}{
\label{line:affinepoints} Augment $\cX$ to contain an affinely independent set of points using \Cref{alg.affpoints}. 
       }

 Call \Cref{alg.improve_poisedness} with $\bar\Lambda$,  $\cX$, and $\bar\Delta_k$ to obtain new $\cX$. \\ \label{line:improve_geometry}

       Obtain evaluations $\tilde{f}(x)$ at any $x\in \cX$ for which we have no evaluation. %

       Solve \cref{eq.MNH} with $\cX$ to obtain $\alpha, \beta$, and construct local model  \cref{eq.model.def.quadratic.alg}.       
       \\ %
        
      Compute $s_k$ as an approximate minimizer of $\{\min m_k(s) : s\in\cB(0, \Delta_k)\} $
      such that $s_k$ satisfies \cref{eq.model.decrease.sufficient}.

 Evaluate $\tilde{f}(\theta_k + s_k)$, augment $\cX \gets \cX\cup\{\theta_k + s_k\}$. \label{line:evaluate_trial} 

Compute 
\begin{equation}
\label{eq.rho2.k}
    \rho_k=\frac{\tilde f(\theta_k)-\tilde f(\theta_k+ s_k)+r\tilde\epsilon_f}{m_k\left(0\right)-m_k\left(s_k\right)}
\end{equation}
         \eIf{$\rho_k \geq \eta_1$ }{
            \begin{flalign*}
            \mbox{Set $\theta_{k+1}=\theta_k+s_k$ and }
                \Delta_{k+1}= \begin{cases}\min\{\gamma^{-1} \Delta_k, \Delta_{\max}\} & \text { if }\left\|g_k\right\| \geq \eta_2 \Delta_k \\ \gamma \Delta_k, & \text { if }\left\|g_k\right\|<\eta_2 \Delta_k\end{cases}&&
            \end{flalign*} \label{line:trust_region_adjustment}
            }{
            Set $\theta_{k+1}=\theta_k$ and $\Delta_{k+1}=\gamma \Delta_k$. \label{line:trust_region_adjustment2}
        }
    }
\caption{Novel noise-aware model-based trust-region Algorithm\label{alg.DFOTR.noise}}
\end{algorithm2e}

\section{Theoretical Results}\label{sec:theory}

 We first demonstrate that, as we intended, the model gradient of a minimum Frobenius norm quadratic model is a first-order oracle for \hyperref[type1]{Type 1} error.

    \begin{lemma}
    \label{lem.bounded_noise}
    Suppose $\tilde{f}$ is a zeroth-order oracle of \hyperref[type1]{Type 1} for $f$ with constant $\epsilon_f$.
    In \Cref{alg.DFOTR.noise}, suppose $\tilde\epsilon_f$ is known exactly in every iteration; that is, $\tilde\epsilon = \epsilon_f$. 
    Additionally, suppose $\tilde{L}_{\nabla f}$ is a valid Lipschitz constant for $\nabla f$ on any given trust region $\cB(\theta_k, \Delta_k)$. 
    Denote the upper bound in \cref{eq.kappa_bmh} by $\kappa_{\text{bmh}}$. 
    Then there exists $A, B > 0$ such that the model gradient $g_k$ is a first-order oracle for $\nabla f(\theta_k)$  with 
    \begin{equation}
    \label{eq.foo.parameters}
    \begin{split}
       & p_1=1, \quad \kappa_{eg}= 
      (p+1)\bar\Lambda (L_{\nabla f} + \max\{A,B\epsilon_f\}), \\  \epsilon_g = (p+1)\bar\Lambda & \max\left\{
      (L_{\nabla f} + A)\sqrt{\frac{r\epsilon_f}{L_{\nabla f}}} + (B+2)\sqrt{\frac{L_{\nabla f}\epsilon_f}{r}}, 
      A\Delta_{\max}^3 + 2\sqrt{\frac{L_{\nabla f}\epsilon_f}{r}}.
      \right\}
    \end{split}
    \end{equation}
    \end{lemma}
    \begin{proof}{\textbf{Proof.}} 
      Combining \Cref{thm.mnh_error}, \Cref{prop.bound_pinv}, and \Cref{thm.model_hessian_bound},  
      we have that, on any iteration $k$ of \Cref{alg.DFOTR.noise}, the model gradient $g_k$ satisfies
      \begin{equation}
      \label{eq.gradient_error}
          \|\nabla f(\theta_k) - g_k\| \leq 
      (p+1)\bar\Lambda\left[(L_{\nabla f} + \|H_k\|) \bar\Delta_k + \frac{2\epsilon_f}{\bar\Delta_k}\right],
      \end{equation}
      where, because we are in the deterministically bounded regime and $\cX \in \cB(\theta_k, \bar \Delta_k)$, we have from \cref{eq.kappa_bmh} that 
      \begin{equation}
      \label{eq.def.A.B}
          \|H_k\| \leq \displaystyle\frac{A}{\nu( \bar \Delta_k)}
+ \displaystyle\frac{B}{ \bar \Delta_k^2\nu( \bar \Delta_k)}\epsilon_f, \quad A:=4(p+1)\bar\Lambda L_{\nabla f}\sqrt{(d+1)(d+2)}, \quad B:= 8(p+1)\bar\Lambda\sqrt{(d+1)(d+2)},
      \end{equation}
for $\nu(\cdot)$ defined in \Cref{def.nu}.   We proceed in three cases.

      \textbf{Case 1:} Suppose $\bar\Delta_k = \Delta_k$ and $\Delta_k \leq 1$. 
      Note that this implies $\bar\Delta_k \geq \sqrt{\frac{r\epsilon_f}{L_{\nabla f}}}$, and so 
      $\bar\Delta_k^{-1} \leq \sqrt{\frac{L_{\nabla f}}{r\epsilon_f}}$. 
      Making use of this fact, we have
$$\begin{array}{rl}
\|\nabla f(\theta_k) - g_k\| \leq & (p+1)\bar\Lambda\left[
(L_{\nabla f} + A)\Delta_k + \frac{(B+2)\epsilon_f}{\Delta_k}
\right]\\
\leq & 
(p+1)\bar\Lambda\left[(L_{\nabla f} + A)\Delta_k + 
(B+2)\sqrt{\frac{L_{\nabla f}\epsilon_f}{r}}
\right] .
\end{array}
$$

\textbf{Case 2:} Suppose $\bar\Delta_k = \Delta_k$ and $\Delta_k > 1$. 
Recall the algorithmic parameter $\Delta_{\max}$, which enforces $\Delta_k \leq \Delta_{\max}$ for all $k$.
Then, 
$$\begin{array}{rl}
\|\nabla f(\theta_k) - g_k\| \leq & (p+1)\bar\Lambda\left[
(L_{\nabla f} + A\Delta_k^2 + B\epsilon_f)\Delta_k + \frac{2\epsilon_f}{\Delta_k}
\right]\\
\leq & 
(p+1)\bar\Lambda\left[
(L_{\nabla f} + B\epsilon_f)\Delta_k + A\Delta_{\max}^3 + \frac{2\epsilon_f}{\Delta_k}
\right] \\
\leq & 
(p+1)\bar\Lambda\left[
(L_{\nabla f} + B\epsilon_f)\Delta_k + A\Delta_{\max}^3 + 2\sqrt{\frac{L_{\nabla f}\epsilon_f}{r}}/
\right] .\\
\end{array}
$$

\textbf{Case 3:} Suppose $\bar\Delta_k = 
\displaystyle\sqrt{\frac{r\epsilon_f}{L_{\nabla f}}}$.
Recall that \Cref{alg.DFOTR.noise} chooses a Lipschitz constant $L_{\nabla f}$  sufficiently large such that $\bar\Delta_k \leq 1$ in this case. 
Then
 $$\begin{array}{rl}
 \|\nabla f(\theta_k) - g_k\| \leq & (p+1)\bar\Lambda\left[
 (L_{\nabla f} + A + \frac{B\epsilon_f}{\bar\Delta_k^2})\bar\Delta_k + \frac{2\epsilon_f}{\bar\Delta_k} 
 \right]\\
 = &
 (p+1)\bar\Lambda\left[
 (L_{\nabla f} + A)\bar\Delta_k + \frac{B+2}{\bar\Delta_k}\epsilon_f
 \right]\\
= &
 (p+1)\bar\Lambda\left[
(L_{\nabla f} + A)\sqrt{\frac{r\epsilon_f}{L_{\nabla f}}} + (B+2)\sqrt{\frac{L_{\nabla f}\epsilon_f }{r}}.
\right].\\
 \end{array}
 $$

The desired result is therefore satisfied with the parameters defined in \cref{eq.foo.parameters}.
      \qed
    \end{proof}

As a brief remark, we note from the analysis in \Cref{lem.bounded_noise} that the maximum in the $\epsilon_g$ term involving the unappealing $\Delta_{\max}^3$ term is  realized only on iterations in which $\Delta_k > 1$.
In fact, $\Delta_{\max}$ can be replaced with $\Delta_k$ on such iterations.
The use of $\Delta_{\max}$ in this analysis is pessimistic but was chosen to make clear how $\epsilon_g$ can be derived as a quantity independent of $\Delta_k$, as is required for the definition of a first-order oracle. 
The next lemma  shows that the model gradient of a minimum Frobenius norm quadratic model is also a first-order oracle for \hyperref[type2]{Type 2} error.

    \begin{lemma}
    \label{lemma.subexponential_noise}
        Suppose $\tilde{f}$ is a zeroth-order oracle for $f$ of \hyperref[type2]{Type 2} with constants $c$ and $\epsilon_f$. 
        Let $\tilde{L}_{\nabla f}$, $\tilde{\epsilon}_f$ and $\kappa_{\text{bmh}}$ be as in \Cref{lem.bounded_noise}. 
        Then there exist $\mu > 1$ and $A, B > 0$ such that the model gradient $g_k$ is a first-order oracle for $\nabla f(\theta_k)$ with 
        \begin{equation}
            \begin{split}
              &  p_2  = 1 - 2(p+1)\exp(c(1-\mu)\epsilon_f), \quad \kappa_{eg} = (p+1)\bar\Lambda (L_{\nabla f} + \max\{A,B\mu\epsilon_f\})\\
      \epsilon_g &= 
      (p+1)\bar\Lambda\max\left\{
      (L_{\nabla f} + A)\sqrt{\frac{r\epsilon_f}{L_{\nabla f}}} + (B\mu+2)\sqrt{\frac{L_{\nabla f}\epsilon_f}{r}}, 
      A\Delta_{\max}^3 + 2\sqrt{\frac{L_{\nabla f}\epsilon_f}{r}}
      \right\} .
            \end{split}
        \end{equation}
    \end{lemma}

    \begin{proof}{\textbf{Proof.}}        By the same reasoning that led to \cref{eq.gradient_error},
        we have that in the $k$th iteration of \Cref{alg.DFOTR.noise},
        $$\|\nabla f(\theta_k) - g_k\| \leq 
      (p+1)\bar\Lambda\left[(L_{\nabla f} + \|H_k\|) \bar\Delta_k + \frac{\epsilon_0 + \epsilon_{\max}}{\bar\Delta_k}\right].$$

      We have only to yield a tail bound on $\epsilon_0 + \epsilon_{\max}$, which we do coarsely
      \footnote{If we cast stronger assumptions about the moment-generating function of $\tilde{f}$, tighter bounds could be yielded via Chernoff bound arguments. 
      To retain the same general definition of zeroth-order oracle that is  stated only in terms of tail bounds, as employed in \citet{Cao2023}, we do not make such assumptions here.}
      with a union bound.
      That is,
      denoting $X_i = |e(x^i,\xi_i)|$ (recalling \cref{eq.error_quantity} for the definition of $e(\theta,\xi)$)
      and letting $\xi$ be a vector of realized random variables $\xi_i$, we have, 
      for any $t\geq 0$,
      $$
      \begin{array}{rl}
      \Prob_{\xi}\left[X_0 + \displaystyle\max_{i=1,\dots,p} X_i  \geq t \right]     
      \leq &
      \Prob_{\xi}\left[X_0 \geq \frac{t}{2}\right]
      +
      \Prob_{\xi}\left[\displaystyle\max_{i=1,\dots,p} X_i \geq \frac{t}{2}\right]\\
      \leq & 
      \exp(c(\epsilon_f - \frac{t}{2})) +
      \Prob_{\xi}\left[\displaystyle\max_{i=1,\dots,p} X_i \geq \frac{t}{2}\right]      \\
      \leq &
      \exp(c(\epsilon_f - \frac{t}{2})) + 
      \displaystyle\sum_{i=1}^p \Prob_{\xi}[X_i > \frac{t}{2}]\\
      \leq &
       (p+1)\exp(c(\epsilon_f - \frac{t}{2})).\\
      \end{array}
      $$
To work with the same quantities as employed in \Cref{lem.bounded_noise}, 
we let $\mu>1$ denote an arbitrary multiplier and consider the special case
$$\Prob_{\xi}\left[X_0 + \max_{i=1,\dots,p} X_i  \geq 2\mu\epsilon_f \right] 
\leq (p+1)\exp(c(1-\mu)\epsilon_f) := p_1(\mu).
$$
      
    Because we are in the independent subexponential noise regime,
    we this time obtain from \cref{eq.kappa_bmh} that
$$\|H_k\| \leq 
\displaystyle\frac{4(p+1)\Lambda L_{\nabla f}\sqrt{(d+1)(d+2)}}{\nu(\bar \Delta_k)}
+ \displaystyle\frac{8\Lambda\sqrt{(d+1)(d+2)}}{\bar\Delta_k^2\nu(\bar\Delta_k)}\sum_{i=0}^p X_i.$$    
      Employing the same union bound argument as before,
      $$
      \Prob_{\xi}\left[
      \sum_{i=0}^p X_i \geq t\right] \leq
      (p+1)\exp\left(c\left(\epsilon_f - \frac{t}{p+1}\right)\right). 
      $$
      Let $\mu > 1$ be the same arbitrary multiplier as before. 
     To work with the same quantities employed in \Cref{lem.bounded_noise}, we are concerned with the tail probability
     $$
      \Prob_{\xi}\left[
      \sum_{i=0}^p X_i \geq \mu(p+1)\epsilon_f\right] \leq
      (p+1)\exp\left(c(1-\mu)\epsilon_f\right) = p_1(\mu). 
      $$
    Now, we can bound
    $$\|H_k\| \leq \frac{A}{\nu(\bar\Delta_k)} + \frac{B\mu}{\bar\Delta_k^2\nu(\bar\Delta_k)}\epsilon_f \quad \text{with probability} \hspace{1pc} p_1(\mu),$$
      where $A, B$ are the same as in \Cref{lem.bounded_noise}.
      
      Proceeding in the same three cases as in \Cref{lem.bounded_noise}, one can almost identically conclude that, for any $\mu > 1$, 
      $g_k$ is a first-order oracle with probability $1 - 2p_1(\mu)$ and constants 
      $\kappa_{eg}= 
      (p+1)\bar\Lambda (L_{\nabla f} + \max\{A,B\mu\epsilon_f\})$ and 
      $\epsilon_g = 
      (p+1)\bar\Lambda\max\left\{
      (L_{\nabla f} + A)\sqrt{\frac{r\epsilon_f}{L_{\nabla f}}} + (B\mu+2)\sqrt{\frac{L_{\nabla f}\epsilon_f}{r}}, 
      A\Delta_{\max}^3 + 2\sqrt{\frac{L_{\nabla f}\epsilon_f}{r}}
      \right\}$. 
      \qed     
    \end{proof}

We now immediately conclude the following convergence results for \Cref{alg.DFOTR.noise}.
\Cref{thm.bounded_noise} and \Cref{thm.subexponential_noise} are direct results of \Cref{theorem.complexity.cao.1} and \Cref{theorem.complexity.cao.2.simple}, respectively, and follow when we explicitly consider how $\epsilon_g$ is a function of $\epsilon_f$ when we employ minimum Frobenius norm modeling. 

  \begin{theorem}
     \label{thm.bounded_noise}Suppose \Cref{ass.L.smooth} and  \ref{ass.f.lower.bounded} hold.
    Suppose $\tilde{f}$ is a zeroth-order oracle of \hyperref[type1]{Type 1} for $f$ with constant $\epsilon_f$.
    In \Cref{alg.DFOTR.noise}, suppose $\tilde\epsilon_f$ is known exactly in every iteration; that is, $\tilde\epsilon_f = \epsilon_f$. 
    Additionally, suppose $\tilde{L}_{\nabla f}$ is a valid Lipschitz constant for $\nabla f$ on a given trust region $\cB(\theta_k, \Delta_k)$. 
Let $\{\Theta_k\}$ denote the sequence of random variables with realizations $\{\theta_k\}$ generated by \Cref{alg.DFOTR.noise}.
There exists $\kappa_1$, independent of $\epsilon_f, \epsilon_g$ but dependent on $\kappa_{eg}, p_1$, such that
given any
$$\epsilon > \kappa_1 \sqrt{\epsilon_f},$$
it holds that
$$\Prob\left[\displaystyle\min_{0\leq k \leq T-1} \|\nabla f(\Theta_k)\| \leq \epsilon\right] 
\geq
1 - \exp\left(-\mathcal{O}(T)\right)$$
for any $T\in\mathcal{O}\left(\frac{1}{\epsilon^2}\right)$.
  \end{theorem}
  \begin{proof}{\textbf{Proof.}} This follows from \Cref{theorem.complexity.cao.1} and \Cref{lem.bounded_noise} with $p_1$, $\kappa_{eg}$, $\epsilon_{g}$ as defined in \Cref{lem.bounded_noise}.
 \qed  
 \end{proof}

  \begin{theorem}
     \label{thm.subexponential_noise}Suppose \Cref{ass.L.smooth} and  \ref{ass.f.lower.bounded} hold.
    Suppose $\tilde{f}$ is a zeroth-order oracle of \hyperref[type2]{Type 2} for $f$ with constant $\epsilon_f$ and $c$.
    In \Cref{alg.DFOTR.noise}, suppose $\tilde\epsilon_f$ is known exactly in every iteration; that is, $\tilde\epsilon_f = \epsilon_f$. 
    Additionally, suppose $\tilde{L}_{\nabla f}$ is a valid Lipschitz constant for $\nabla f$ on a given trust region $\cB(\theta_k, \Delta_k)$. 
Let $\{\Theta_k\}$ denote the sequence of random variables with realizations $\{\theta_k\}$ generated by  \Cref{alg.DFOTR.noise}.
There exists $\kappa_2$, independent of $\epsilon_f, \epsilon_g$ but dependent on $\kappa_{eg}, p_1, c$, such that
given any
$$\epsilon > \kappa_2 \sqrt{\epsilon_f + 1/c},$$
it holds that
$$\Prob\left[\displaystyle\min_{0\leq k \leq T-1} \|\nabla f(\Theta_k)\| \leq \epsilon\right] 
\geq
1 - 2\exp\left(-\mathcal{O}(T)\right)- \exp\left(-ct/4\right) $$
for any $T\in \mathcal{O}\left(\frac{1}{\epsilon^2}\right)$ and any $t > 0$.  
  \end{theorem}
  \begin{proof}{\textbf{Proof.}} This follows from \Cref{theorem.complexity.cao.2.simple} and \Cref{lemma.subexponential_noise}, with $p_1$, $\kappa_{eg}$ and $\epsilon_{g}$ as defined in \Cref{lemma.subexponential_noise}.   
 \qed  
 \end{proof}

\section{Numerical Results}\label{sec:results}
We now study the practical performance of an implementation of \Cref{alg.DFOTR.noise} and other methods commonly used for optimizing VQAs.    
\subsection{\texttt{ANATRA}}
We have created a reasonably faithful Python implementation of \Cref{alg.DFOTR.noise}. We refer to this implementation as \texttt{ANATRA}, short for  A Noise-Aware Trust-Region Algorithm. \texttt{ANATRA} differs from \Cref{alg.DFOTR.noise} in a few minor ways, which we now list.
\begin{enumerate}
\item In \Cref{line:improve_geometry} we do not actually run \Cref{alg.improve_poisedness} until a $\bar\Lambda$-poised set is returned. 
Instead, if $\cX$ is not already $\bar\Lambda$-poised, then we run a \emph{single iteration} of \Cref{alg.improve_poisedness} and return an improved, but not necessarily $\bar\Lambda$-poised, interpolation set $\cX$.
With this modification, \Cref{alg.improve_poisedness} additionally returns a validity flag that is set to True provided the set $\cX$ returned by \Cref{alg.improve_poisedness} is indeed $\bar\Lambda$-poised.  
\item We do not evaluate $\tilde{f}(\theta_k + s_k)$ in \Cref{line:evaluate_trial} unless at least one of the following two conditions holds: either the validity flag returned in \Cref{line:improve_geometry} is True, or the trial step satisfies $\|s_k\| \geq 0.01\Delta_k$. \replace{}{When $\tilde{f}(\theta_k + s_k)$  is not evaluated, we set $\Delta_{k+1}\gets\Delta_k$ and $\theta_{k+1}\gets\theta_k$ and return to \Cref{line:first_step}.}

\item \Cref{alg.TR.noise} was analyzed with a trust-region adjustment step like that in \Cref{line:trust_region_adjustment}.
However, we found it more practical to use the following trust-region adjustment step instead:
If $\rho_k\geq \eta_1$ and $\|s_k\| > 0.75\Delta_k$, then $\Delta_{k+1} \gets \gamma^{-1}\Delta_k$. Otherwise (if $\rho_k < \eta_1$), then $\Delta_{k+1} \gets \gamma\Delta_k$ only provided the validity flag is currently set to True.
We remark that both this change and the previous change are also employed in \texttt{POUNDers} \cite{wild2017chapter}.

\item While the analysis of \Cref{alg.TR.noise} certainly thrives on the nonmonotonicity of noisy function values, as evident in the ratio \cref{eq.rho.k}, we found allowing too much nonmonotonicity can deter practical performance. 
In particular,  we maintain a memory of the lowest noisy function value $\tilde{f}(\theta_{best})$ observed up until iteration $k$ (and its corresponding $\theta_{best}$), which we will denote $\tilde{f}_{best}$. 
If, at the end of the $k$th iteration of \Cref{alg.DFOTR.noise},
$\tilde{f}(\theta_k) \geq \tilde{f}_{best} + r \epsilon_f$,
then we replace the incumbent $\theta_k$ with the stored $\theta_{best}$. While this goes against the theory of \Cref{alg.TR.noise}, and hence \Cref{alg.DFOTR.noise}, this safeguard against \emph{too much} non-monotonicity appears to greatly aid practical performance.  
\end{enumerate}

In terms of parameters for \Cref{alg.DFOTR.noise}, we use values that are considered fairly standard in model-based trust-region methods. 
In particular, we use $\eta_1 = 0.25$, $\gamma = 0.5$, $c_s = \sqrt{d}$, and $\bar\Lambda = \sqrt{d}$. 
We note that owing to our third adjustment above, there is no $\eta_2$ in our method. 
In terms of the parameter that is unique to an algorithm within the framework of \Cref{alg.TR.noise}, we set $r=2$ in \Cref{alg.DFOTR.noise}.

We remark that  a lot of freedom exists in how \Cref{line:get_noise_estimate} is performed. In different numerical tests we will have different methods of obtaining $\tilde\epsilon_f$.
In all of our numerical tests, however, we will always compute $\tilde{L}_{\nabla f}$ in the same way. 
In the first iteration ($k=0$) of \Cref{alg.DFOTR.noise}, we simply let $\tilde{L}_{\nabla f} = 1$. 
Otherwise, 
in any iteration where the validity flag is True at the time \Cref{line:get_noise_estimate} is reached, we update $\tilde{L}_{\nabla f}$ to be the largest eigenvalue of the most recent model Hessian $H_{k-1}$. 
Otherwise, if the validity flag is False when \Cref{line:get_noise_estimate} is reached, we do not update the estimate $\tilde{L}_{\nabla f}$. 

\subsection{Comparator Methods}
In our tests it is naturally appropriate to compare \texttt{ANATRA} with the various solvers wrapped in \texttt{scikit-quant} \citep{lavrijsen2020classical}.
In particular, we choose to compare \texttt{ANATRA} with \texttt{PyBOBYQA} \citep{cartis2019improving}, \texttt{NOMAD} \citep{LeDigabel2011}, and \texttt{ImFil} \citep{imfil}, each as wrapped in \texttt{scikit-quant}. 
Of the four \texttt{scikit-quant} solvers discussed in \citet{lavrijsen2020classical}, we are noticeably missing \texttt{SnobFit} \citep{Huyer2008}. 
This omission is intentional; whereas the other three solvers are local methods, \texttt{SnobFit} was designed for global optimization.
\texttt{SnobFit} constructs local quadratic models of the objective function centered at points selected within a branch-and-bound framework. 
As such, the performance of \texttt{SnobFit} critically depends on the choice of a set of bound constraints. 
Because we are testing the performance of optimization methods for \emph{unconstrained} problems, bound constraints were not provided to any solver.
Preliminary tests indeed showed, as we expected, that the performance of \texttt{SnobFit} was very sensitive to the size of provided bound constraints.  
One should not discount the utility of \texttt{SnobFit} based on these tests, but a practitioner should be aware that \texttt{SnobFit} is only a reasonable comparator in settings where a meaningfully bounded domain can be provided by the user. 

\paragraph{\texttt{ImFil}}
Implicit filtering \citep{imfil} is a method with a fairly sophisticated implementation (which we denote \texttt{ImFil}).
At its core, implicit filtering is an inexact quasi-Newton method.
Gradient approximations are constructed via central finite differences with initially coarse difference parameters.
On any iteration in which the center of the finite difference stencil exhibits a lower objective value than all of the forward or backward evaluation points, a \emph{stencil failure} is declared, and the finite difference parameter is decreased. 
The original motivation for implicit filtering came from noisy problems; in fact, the convergence of implicit filtering to local minimizers can be established when the objective function is the sum of a ground truth smooth function $f_s$ and a high-frequency low-amplitude noisy function $f_n$ provided $f_n\to 0$ at local minimizers of $f_s$ \citep{Gilmore1995}.
While such an assumption may not be precisely satisfied by our problem setting, \texttt{ImFil} remains a competitive method in the presence of more general noise settings, including ours.  

\paragraph{\texttt{NOMAD}}
\texttt{NOMAD} \citep{LeDigabel2011} is a widely used implementation of a mesh-adaptive direct search algorithm \citep{Audet2006}. 
Unlike model-based methods (including all of the methods discussed so far), direct search methods never explicitly construct a model of a gradient or higher-order derivatives. 
Instead, a mesh-adaptive direct search (MADS) algorithm maintains an incumbent point and on each iteration generates a set of points on a(n implicit) mesh covering the feasible region.
If sufficient objective improvement is found on a generated mesh point, the incumbent is updated. 
Because there is so much freedom in how points are generated on the mesh and how the mesh is updated between iterations and because any number of heuristics can be inserted into a so-called search step of MADS, implementation details of MADS algorithms are critical.
While MADS algorithms, and in particular \texttt{NOMAD}, were not designed for noisy problems,  an intuitive argument can be made that direct search methods could be slightly more robust to noise than are model-based methods, since they cannot be affected by poor (interpolation) models that have accumulated too many noise function evaluations. 
Recent work has seen extension of MADS algorithms to stochastic problems \citep{audet2021stochastic}, but these extensions are not currently part of \texttt{NOMAD}.

\paragraph{\texttt{PyBOBYQA}}
\texttt{PyBOBYQA} 
is an extension of \texttt{BOBYQA}, a widely used model-based derivative-free solver \citep{powell2009bobyqa}. 
Because \texttt{ANATRA} also belongs to the class of model-based derivative-free solvers, \texttt{PyBOBYQA} is the closest in spirit to \texttt{ANATRA} among all the comparator methods in this study. 
Of particular note, \texttt{PyBOBYQA} implements several heuristics designed to make \texttt{BOBYQA} more robust to noise, which is of critical importance in our setting.  
This robustness is accomplished primarily via a \emph{restart} mechanism.
\texttt{PyBOBYQA} stores the norm of model gradient and Hessian norms over several past iterations. If an exponentially increasing trend is detected over that short history, 
then the trust-region radius is increased, and several interpolation points, including the trust-region center, are selectively replaced.  

\paragraph{\texttt{SPSA}} 
In addition to the discussed solvers wrapped by \texttt{scikit-quant}, we  test the method of simultaneous perturbation by stochastic approximation (\texttt{SPSA}) \citep{Spall1992}. 
Conceptually, \texttt{SPSA} is a remarkably simple method that computes a coarse stochastic gradient approximation by computing a two-point finite-difference gradient estimate.
At first blush, two-point gradient estimates are immensely attractive, because it appears that considerable optimization progress can be made even when the dimension of the problem is large; contrast these two evaluations with any model-based method that effectively requires $\mathcal{O}(d)$ many function evaluations to even begin the optimization. 
While asymptotic convergence to local minima can be established for \texttt{SPSA} \citep{Gerencsr1997, Kleinman1999}, 
various parameters, including the step-size sequence, can be difficult to tune in practice. 
Despite this, \texttt{SPSA} has become popular within the quantum computing community. A note on the Qiskit documentation \cite{QISKIT-SPSA-2024} states: 

``\texttt{SPSA} can be used in the presence of noise, and it is therefore indicated in situations involving measurement uncertainty on a quantum computation when finding a minimum. If you are executing a variational algorithm using an OpenQASM simulator or a real device, \texttt{SPSA} would be the most recommended choice among the optimizers provided here.''

For these reasons, we have chosen to include Qiskit's implementation of \texttt{SPSA} among the tested solvers. We note that, as with all of the solvers tested in our experiments, we only use default settings, and in particular we do not attempt to tune the learning rate of \texttt{SPSA}. 
The intention behind this choice is to provide further evidence for the well-known empirical observation that trust-region methods (such as \texttt{ANATRA} and \texttt{PyBOBYQA}) are relatively insensitive to hyperparameter selections, when compared with methods based on stochastic approximation, such as \texttt{SPSA}. 

\subsection{Tests on Synthetic Problems}
We strongly believe that an important numerical test of a model-based derivative-free method in the presence of noise is arguably the simplest one imaginable: a quadratic function perturbed by additive noise. 
In particular, for a given dimension $d$ so that $\theta\in\Reals^d$,
\begin{equation}
\label{eq:synthetic_quadratic}
\tilde f(\theta, \xi) = \theta^\top \theta + \xi.
\end{equation}
The objective function in \cref{eq:synthetic_quadratic} is an ideal test function for these methods because, if there were no noise, any derivative-free method that attempts to build a quadratic interpolation model ought to construct a perfect (that is, the quadratic interpolation model exactly equals the objective function) local model of the objective function as soon as $2d+1$ function evaluations are performed on a set of points exhibiting reasonable geometry.
In our tests, in an effort to test both deterministically bounded noise regimes and independent subexponential noise regimes, we let 
$\xi$ in \cref{eq:synthetic_quadratic} be uniformly drawn at random from $[-\epsilon_f,\epsilon_f]$ or else $\xi\sim\mathcal{N}(0,\epsilon_f^2)$, respectively for a noise level $\epsilon_f$ that we choose. 
For the synthetic tests, we explicitly provide \texttt{ANATRA} with the chosen value of $\epsilon_f$ as an input. 
Obviously, in our real problems, we will not have this luxury, but we aim to demonstrate in the synthetic tests how well \texttt{ANATRA} does given idealized estimates. \replace{}{\texttt{ANATRA} and \texttt{PyBOBYQA} both require a construction of a local quadratic model in each iteration and therefore have relatively large per-iteration linear algebra cost as measured in CPU time compared to other solvers, especially when the dimension is large.  
  However, the main computational cost in our problems of interest is the number of function evaluations. 
  With this perspective, the CPU time used for model construction in both algorithms is negligible. Consequently, we present performance plots solely with respect to the number of function evaluations. }

Median performance, over 30 trials for each solver, is illustrated for $d=2$ in \Cref{fig:2dquadratic} and for $d=10$ in \Cref{fig:10dquadratic}.
In these tests we chose $\theta_0$ as the vector of all ones, and we test on noise levels $\epsilon_f\in\{10^{-5},10^{-3},10^{-1}\}$. 
In general, we observe a clear preference for \texttt{ANATRA} except 
for the lowest level of noise ($\epsilon_f = 10^{-5}$) and on the larger problem ($d=10$), in which 
\texttt{PyBOBYQA} finds better-quality solutions within the budget of $25(d + 1)$ function evaluations. 
We note that the relative preference for using \texttt{PyBOBYQA} decreases as the noise increases. 
This was to be expected, since as this paper further demonstrates, the quality of a quadratic interpolation models become proportionally poor as $\epsilon_f$ in a noise oracle increases, and the noise mitigation technique employed by \texttt{PyBOBYQA} is only a heuristic.  

\begin{figure}
    \centering
    \includegraphics[width=.9\textwidth]{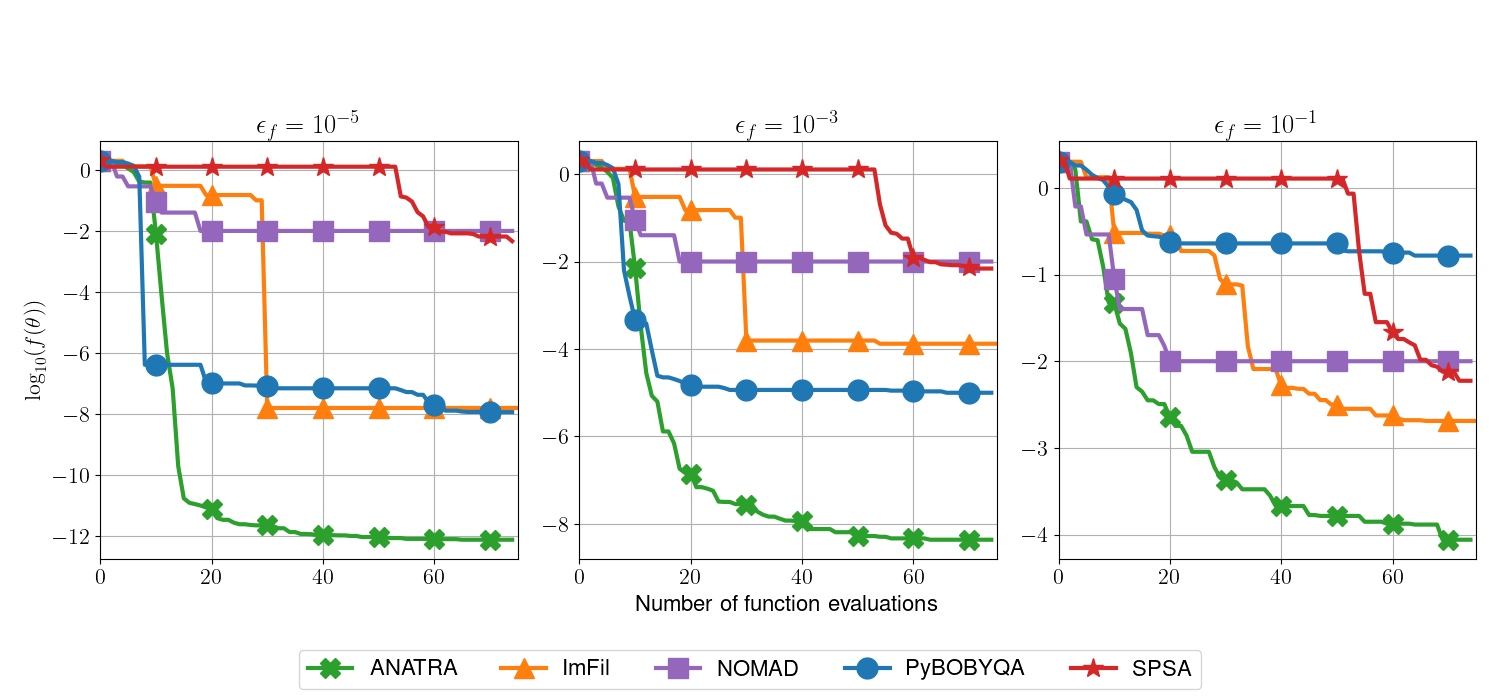}
    
        \includegraphics[width=.9\textwidth]{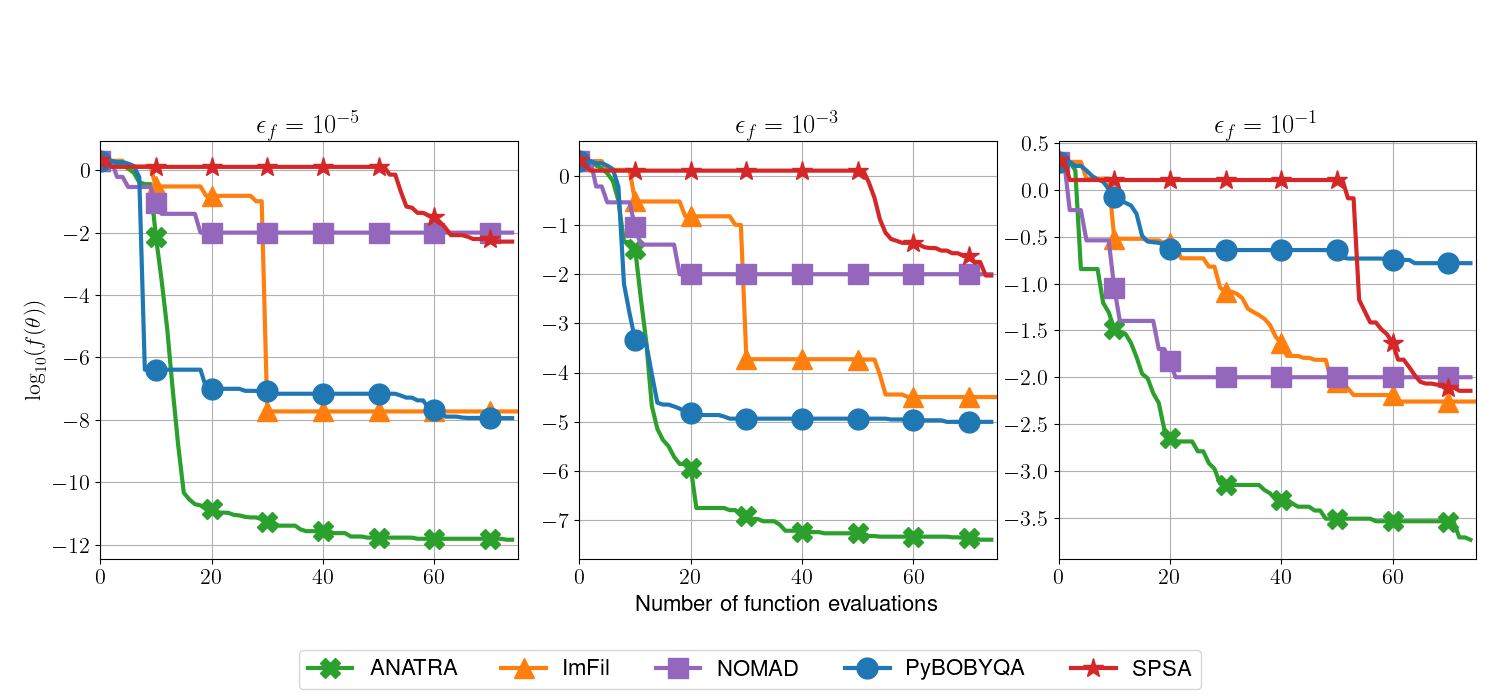}

        \caption{\label{fig:2dquadratic}Results for $d=2$-dimensional quadratic problems \cref{eq:synthetic_quadratic}. 
        Top row corresponds to uniform noise; bottom row corresponds to Gaussian noise. 
        Throughout all of these plots, the solid lines correspond to the median ground truth objective value over 30 runs of the best point evaluated by the solver.  }
\end{figure}

\begin{figure}
    \centering
 \includegraphics[width=.9\textwidth]{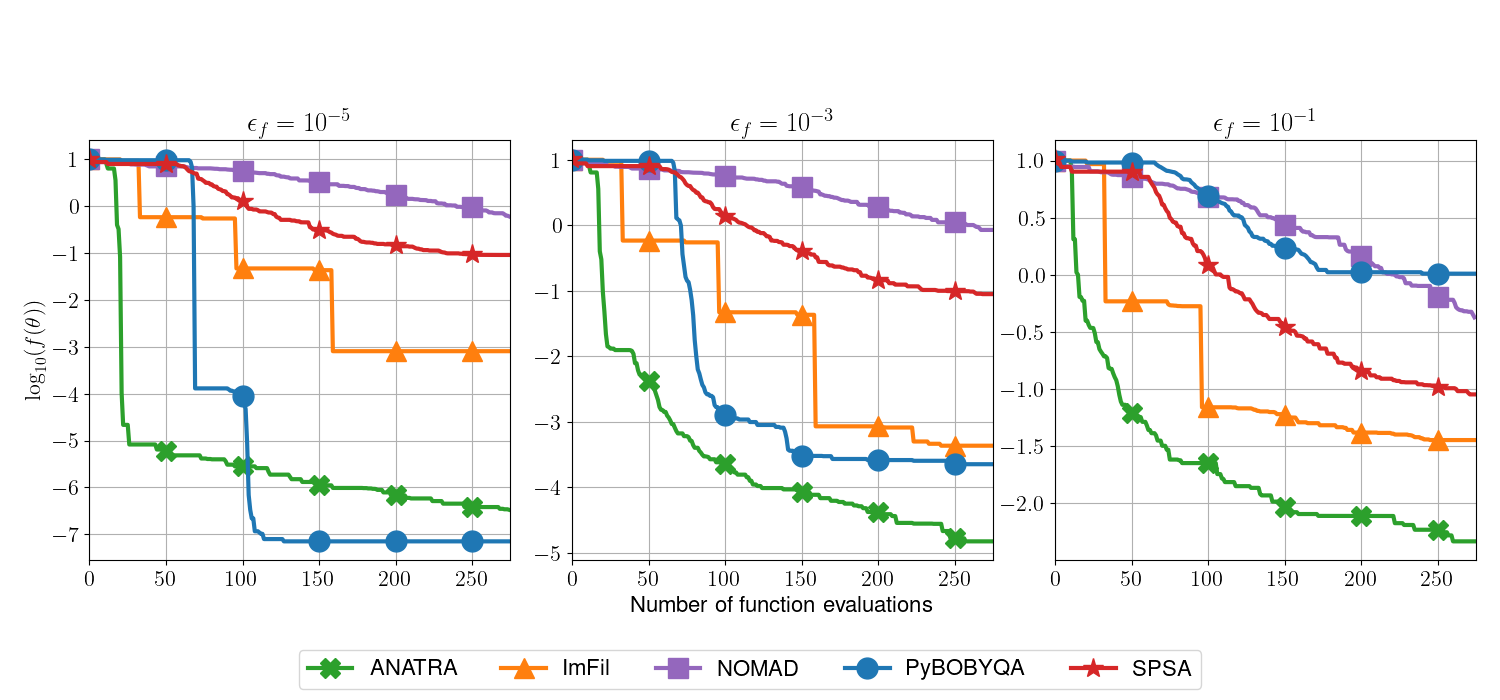}
 
        \includegraphics[width=.9\textwidth]{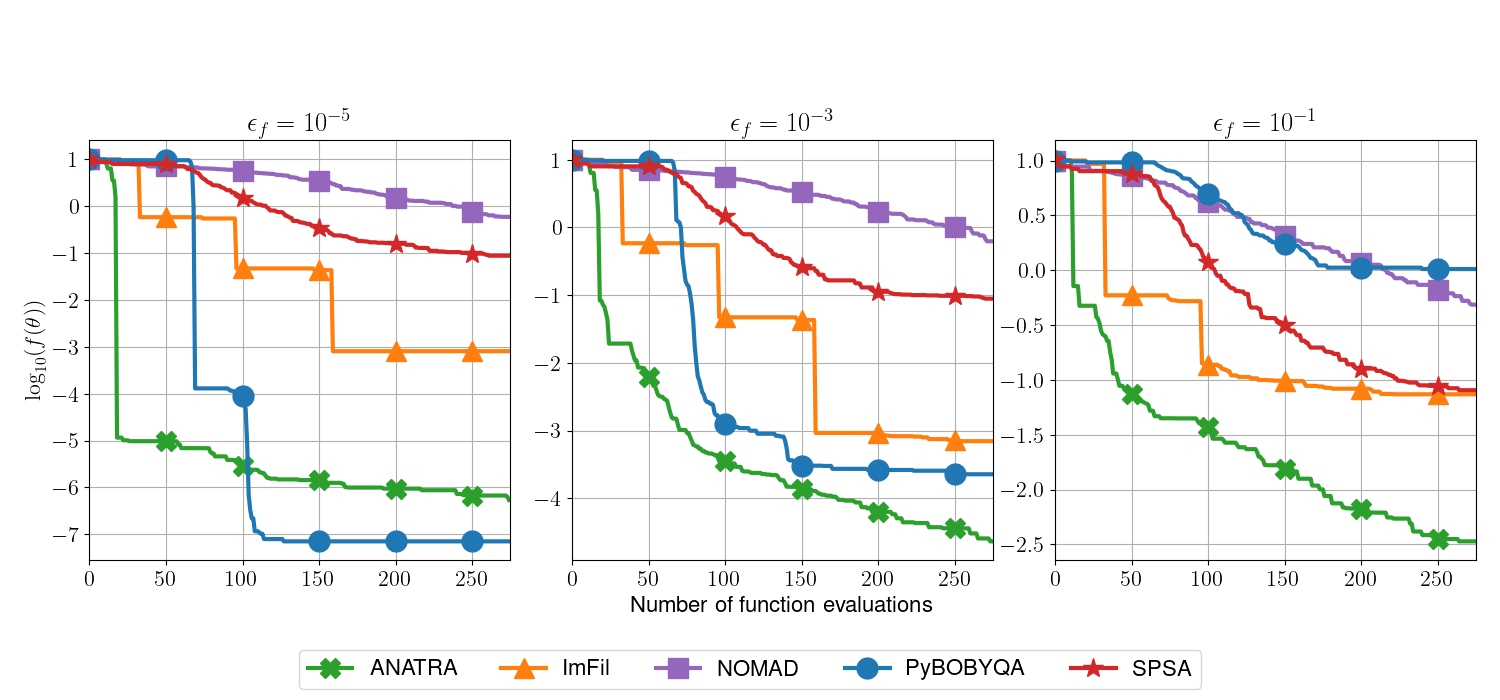}

         \caption{\label{fig:10dquadratic}Same as \Cref{fig:2dquadratic} but with $d=10$-dimensional quadratic functions.}
\end{figure}       

A second synthetic problem that we find important for comparing derivative-free optimization methods is the standard benchmark 2-dimensional Rosenbrock function perturbed by additive noise.
That is, 
\begin{equation}
    \label{eq:synthetic_rosenbrock}
   \tilde f(\theta,\xi) = 100(\theta_2 - \theta_1^2)^2 + (1-\theta_1)^2 + \xi.
\end{equation}
The Rosenbrock function is especially good for testing the efficacy of \emph{model-based} derivative-free methods because the Rosenbrock function is a quartic polynomial, meaning that a quadratic interpolation model will generally never be a perfect model. 
Moreover, the Rosenbrock function is highly nonlinear but interpretably so; any descent-seeking trajectory must turn around a curved valley, the base of which is defined by the curve $\theta_2 = \theta_1^2$. 
However, even as a pathologically nonlinear and nonconvex function, the 2-dimensional Rosenbrock has exactly one local (global) minimum, making benchmarking in terms of function values straightforward. 

Results comparing the median performance over 30 trials of the five solvers are displayed in \Cref{fig:rosenbrock}. 
In these runs we used the starting point of the origin, which is conveniently on the curve $\theta_2= \theta_1^2$; that is, this test is designed not to test an algorithm's ability to find the bottom of the valley but instead  to test an algorithm's ability to follow the nonlinear valley to the global minimum. 
We note a few behaviors that did not appear in the test with the simple quadratic functions. 
First, one may note that \texttt{SPSA} does not seem to start from the same starting point as the other solvers; the reason that  because of the finite differencing scheme innate to \texttt{SPSA}, the initial point is never actually evaluated. 
Moreover, because the gradient is relatively small near $\theta_2=\theta_1^2$ but relatively large farther away from the same curve, the gradient estimates obtained from the randomized two-point difference scheme employed by \texttt{SPSA} lead to a trajectory that tends to stay near the valley but never gets too close to the bottom until/unless a very small step size is employed. 
Of the remaining solvers, it is notable that a relative preference for \texttt{NOMAD} versus \texttt{ImFil} seems to have switched for this problem. 
A potential explanation for this may lie in the fact that \texttt{ImFil} uses a fairly rigid (coordinate-aligned) finite difference gradient estimator; this stencil centered near any point on $\theta_2=\theta_1^2$ will not overlap well with the valley of descent. 
This will trigger multiple stencil failures until the stencil size is quite small, at which point finding descent is difficult in the presence of noisy evaluations. 
\texttt{NOMAD}, on the other hand, generates polling directions less rigidly on the mesh and is more likely to identify an improving point. 
Of particular interest to this paper, \texttt{ANATRA} and \texttt{PyBOBYQA} perform similarly on the lowest level of noise ($\epsilon_f=10^{-5}$), but the preference for using \texttt{ANATRA} becomes increasingly clear as the noise increases.

\begin{figure}
    \centering
    \includegraphics[width=.9\textwidth]{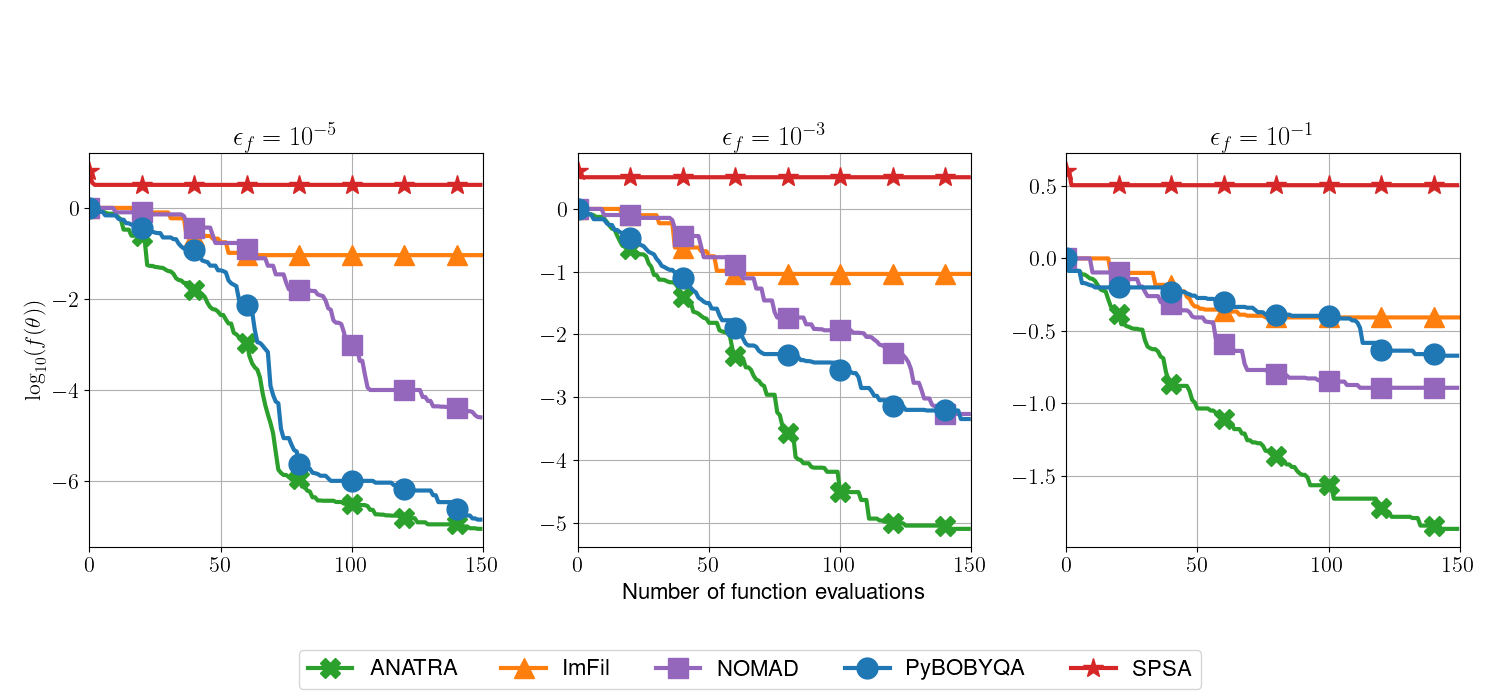}
    
        \includegraphics[width=.9\textwidth]{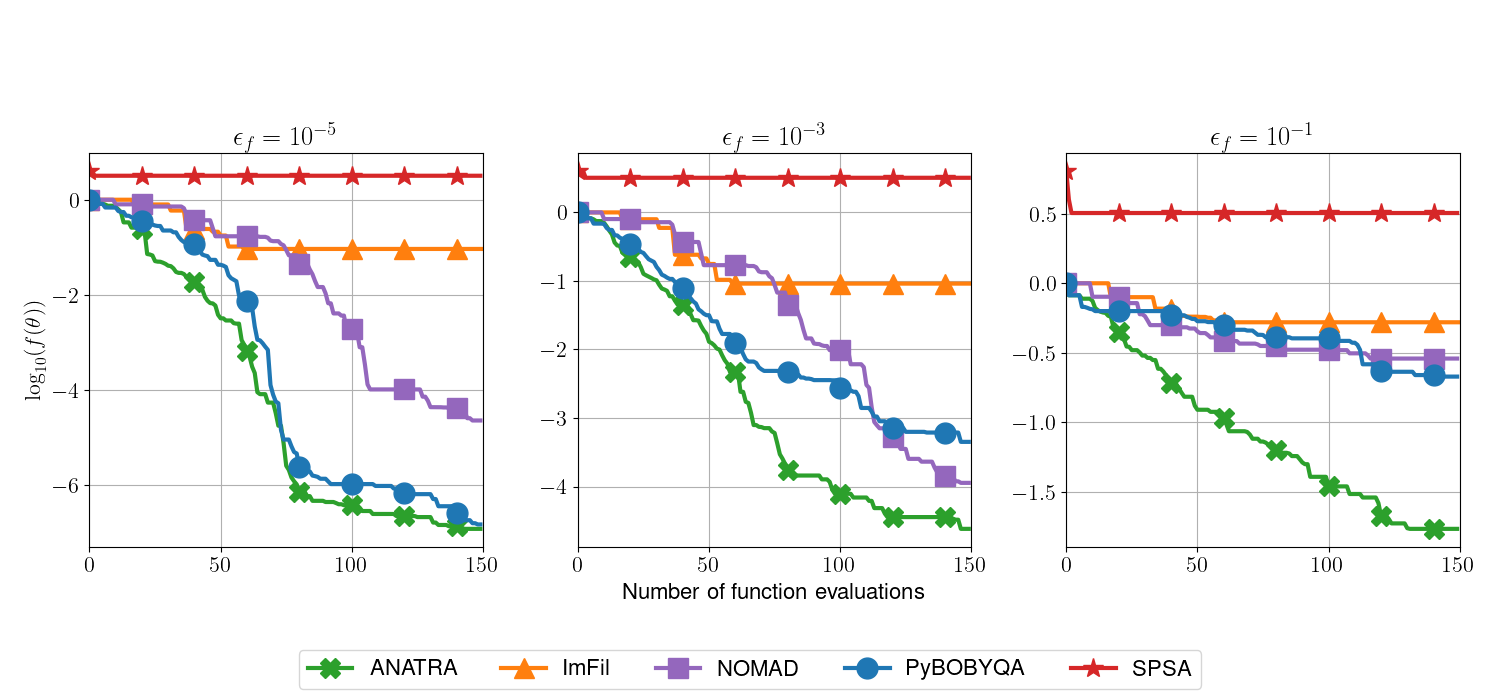}
        \caption{\label{fig:rosenbrock} Comparing solvers on \cref{eq:synthetic_rosenbrock}. Top plots correspond to uniform noise; bottom plots correspond to Gaussian noise. }
\end{figure}

\subsection{Tests on VQA Problems}
The synthetic tests of the preceding section were designed to establish why we believe a noise-aware model-based method like \texttt{ANATRA} is a good choice for noise-perturbed smooth optimization. 
To further this claim,
we now perform tests on standard QAOA benchmarks to demonstrate the performance of our solver on simulations of real problems, which is the original motivation for our work.  
We simulate QAOA MaxCut circuits in Qiskit~\citep{Qiskit} with a depth of five, resulting in a set of (ten) parameters. 
Of course, to mirror the real-world, we no longer assume that we know $\epsilon_f$ as an input to \texttt{ANATRA}.
Instead, when we compute the sample average of the MaxCut objective values suggested by the quantum device, we additionally compute the standard error; we use the standard error as  $\tilde\epsilon_f$.  
In our tests we
vary the shot counts per function evaluation to be in $\{50, 100, 500, 1000\}$. 

We experiment with the MaxCut problem  both on a toy graph 
with MaxCut value of 6 and 
on the Chv\'{a}tal graph, a standard benchmark that has a MaxCut value of 20. 
\replace{}{In our first set of QAOA experiments, illustrated in \Cref{fig:qaoa}, we employed the QASM simulator in Qiskit to simulate an ideal execution of the QAOA circuit.}

\begin{figure}
    \centering
    \includegraphics[width=.9\textwidth]{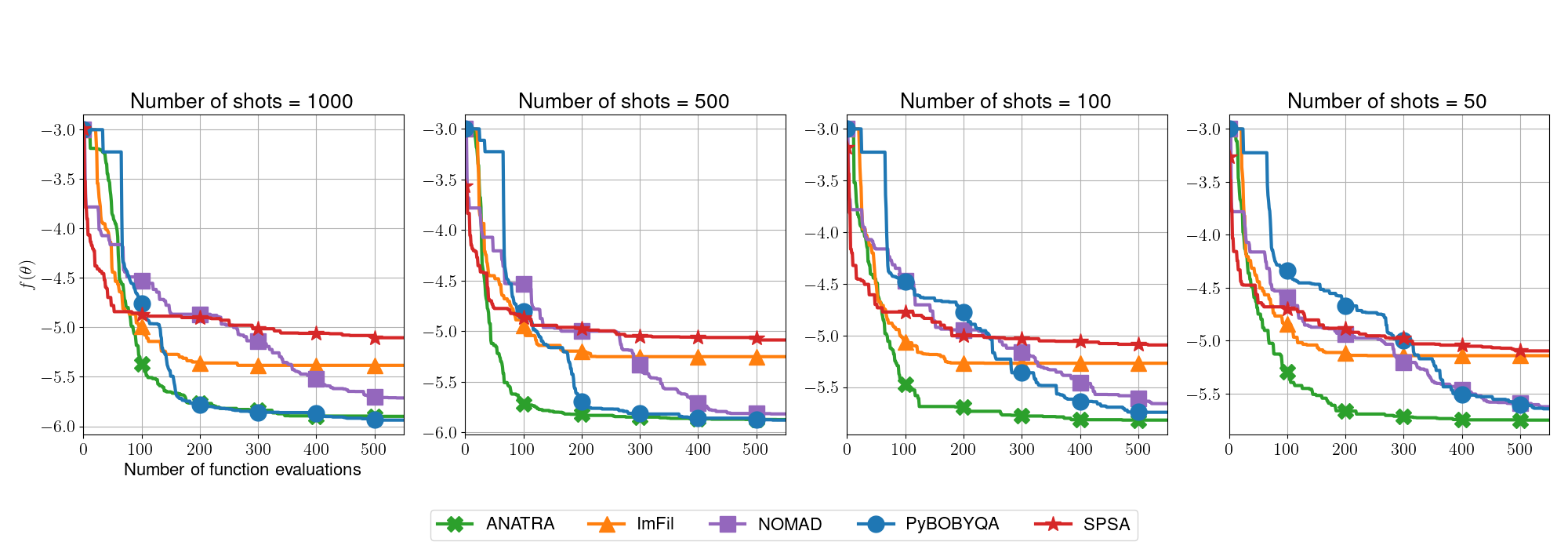}
    
        \includegraphics[width=.9\textwidth]{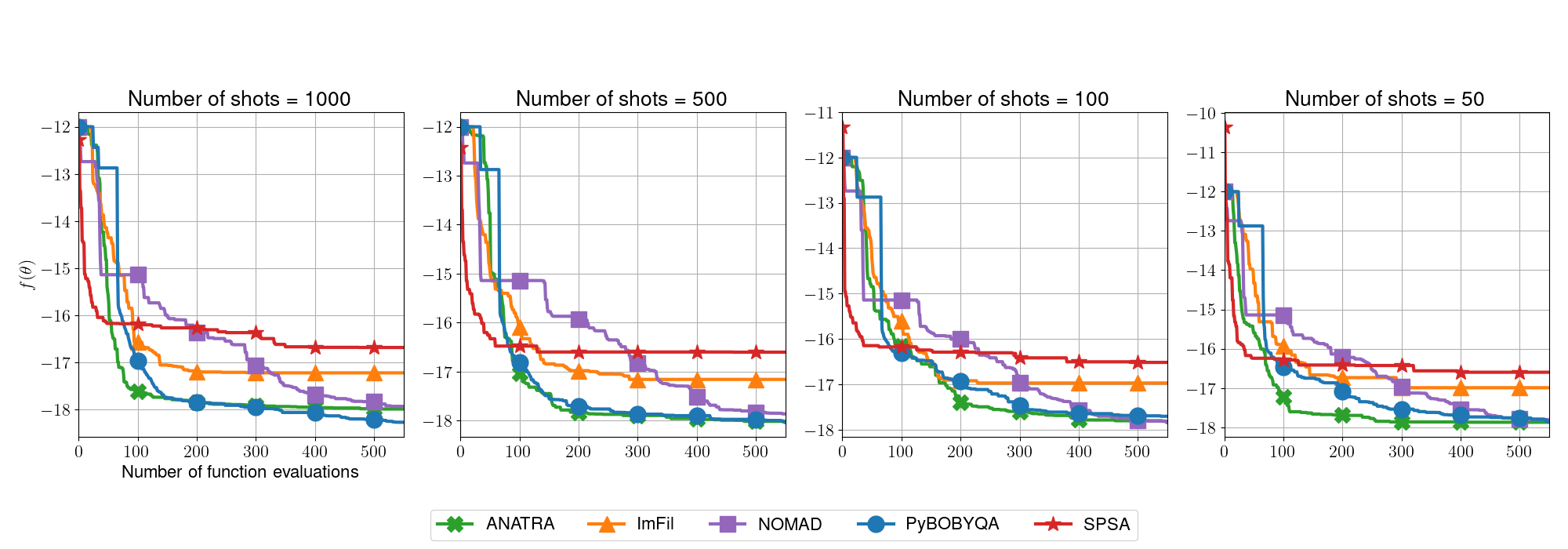}
        \caption{\label{fig:qaoa} Comparing solvers on QAOA MaxCut problems. Top plots correspond to the toy graph, and bottom plots correspond to the Chv\'{a}tal graph.}
\end{figure}

The results in \Cref{fig:qaoa} mirror most of our expectations that came from the synthetic tests. In less noisy settings (when the shot count is 1000 shots per evaluation), there is little distinction, but perhaps a slight preference, for using \texttt{PyBOBYQA} over \texttt{ANATRA}. However, as the noise increases (the shot count decreases), we see an increasingly clear preference for employing \texttt{ANATRA}, both in terms of final median solution quality and in terms of efficiency to reach said median solution quality. 

\replace{}{The practical intention of \texttt{ANATRA} was to provide a noise-aware solver, where the source of the noise need not be purely stochastic, but rather real gate noise on near-term quantum devices. 
As such, we also experimented with \texttt{ANATRA} on the smaller (toy graph) QAOA MaxCut instance on IBM Algiers.
Due to time and budget constraints, we only ran this second set of experiments with the two best-performing solvers from the tests illustrated in \Cref{fig:qaoa}, \texttt{PyBOBYQA} and \texttt{ANATRA}. 
Although not entirely technically correct, it is reasonable to continue using the standard error of cost function estimates to estimate $\epsilon_f$ in the presence of gate noise. 
Developing more robust means of estimating $\epsilon_f$ to better align with \Cref{def:zoo} is a topic for future research. 
Median results are displayed in \Cref{fig:qaoa2}. 
We comment that while both \texttt{PyBOBYQA} and \texttt{ANATRA} perform worse in this noisier environment as compared to the idealized QASM simulator, the relative advantage of using \texttt{ANATRA} over \texttt{PyBOBYQA} has clearly increased in this noisier setting. 
}

\begin{figure}
    \centering
    \includegraphics[width=.9\textwidth]{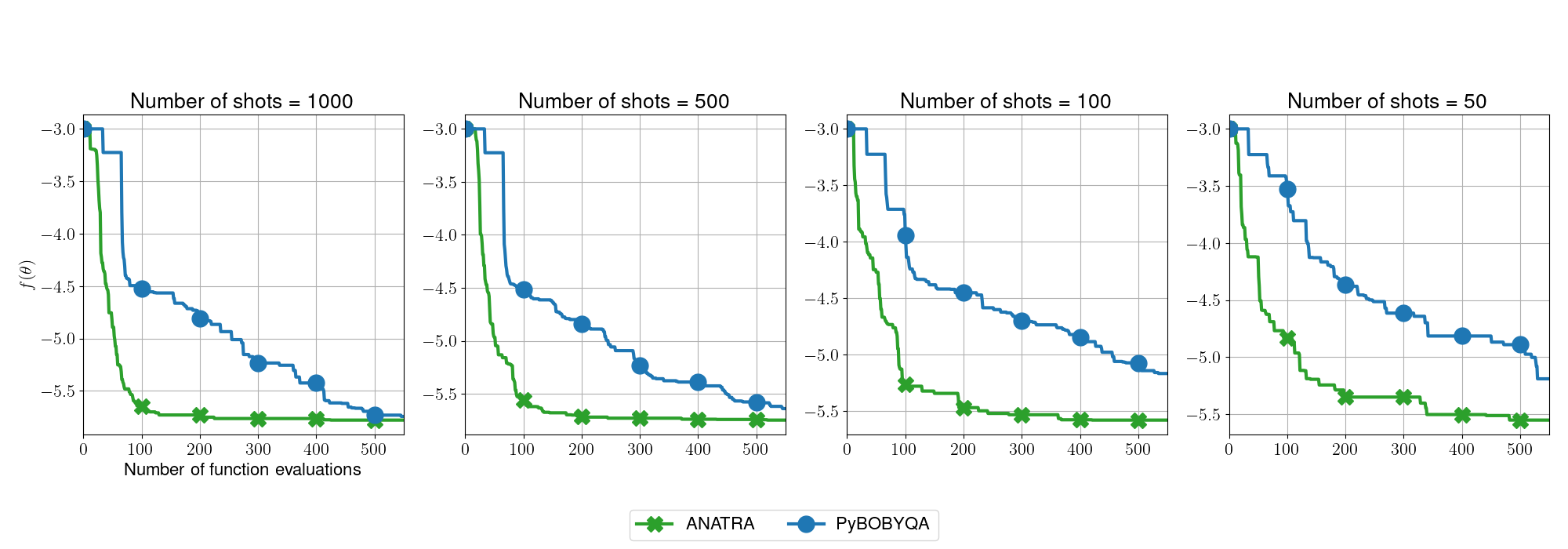}
    
    \caption{\label{fig:qaoa2} Comparing the performance of \texttt{ANATRA} and \texttt{PyBOBYQA} on the QAOA MaxCut problems with the toy graph, but on IBM Algiers, as opposed to an idealized QASM simulator. }
\end{figure}

\section{Discussion}\label{sec:discussion}

We have presented, analyzed, and tested a noise-aware model-based trust-region algorithm to solve noisy derivative-free optimization problems, a problem class that can encompass VQA.
In our theory, function evaluations are assumed to be obtained from a zeroth-order oracle with deterministically 
 bounded noise or  subexponential noise.
 Our proposed algorithm was based on an established noise-aware trust-region method but employed algorithmic devices to carefully maintain poisedness of interpolation points.
  In addition, unlike most classical model-based trust-region methods, our method decoupled the trust-region radius from the sampling radius.
 These two considerations were made in order to guarantee conditions concerning first-order oracles, required by the theory of the original noise-aware trust-region method, were satisfied. 
 Building on previous results, 
 we proved that with high probability our method exhibits a worst-case $\mathcal{O}(\epsilon^{-2})$ convergence rate to an $\epsilon$-neighborhood of a local minima provided $\epsilon$ is greater than a function of the noise level $\epsilon_f$. 
 Numerical experiments demonstrate that our proposed algorithm outperforms alternative solvers, particularly in highly noisy regimes, such as when shot counts on a quantum device are low. 

The work in this manuscript leaves open several avenues for future development. 
As mentioned, the techniques proposed in \citet{Powell2003} could alleviate the considerable per-iteration linear algebra overhead incurred by \texttt{ANATRA}.
Negotiating between the theory and practice of these model update procedures involves nontrivial research effort. 
We are also interested in extending \texttt{ANATRA} with adaptive sampling techniques appropriate for stochastic optimization, such as those employed in \texttt{ASTRO-DF} \citep{Shashaani2018}. 
While the assumptions made in \Cref{def:zoo} did not require oracles to be an unbiased estimator of a ground truth function, there is a significant stochastic component to each noisy function evaluation done on a quantum computer.
As long as the hardware error does not dominate the stochastic error, stochastic optimization (and hence adaptive sampling) may be appropriate. 
Opportunities exist for developing techniques to distinguish between stochastic and hardware noise, so that adaptive sampling may be effectively and judiciously performed in the VQA setting; a step towards such judicious detection can be found in \cite{menickelly2024estimating}.
The application of \texttt{ANATRA} to other noise settings is also of interest. 
For example, \Cref{def:zoo} established a \emph{global} property for oracles in the sense that $\epsilon_f$ was a constant applicable to all of $\Reals^d$. 
However, as can be seen in \Cref{line:get_noise_estimate} of \Cref{alg.DFOTR.noise}, \texttt{ANATRA} was designed assuming $\epsilon_f$ is in fact a \emph{local} constant,  intended to be relevant only on a trust region in each iteration. 
In problem settings where noise is known to be nonconstant with respect to problem parameters, such as many VQA settings (see, e.g., \citet{zhang2022variational} and references therein), a potential extension of \texttt{ANATRA} might attempt to model nonconstant noise.
This can be done, for instance, by constructing interpolation/regression models of $\epsilon_f$ and employing the resulting noise model not only to perform \Cref{line:get_noise_estimate} but also to modify routines for selecting $\cX$ to decrease the error of noisy model gradients.

\section*{Acknowledgments}
This material is based upon work supported by the U.S.
Department of Energy, Office of Science, National Quantum Information Science
Research Centers and the Office of Advanced Scientific Computing Research,
Accelerated Research for Quantum Computing program under contract number
DE-AC02-06CH11357.

\begin{appendices}
    \section{Algorithm for generating affinely independent points}
    Here we present an algorithm that is used in \Cref{line:affinepoints} of \Cref{alg.DFOTR.noise}. 
    This algorithm is based on \cite{wild2008mnh}[Algorithm 4.1]. 
    \Cref{alg.affpoints} begins by computing the set of displacements of each point in $\cX$ from the center point, $x^0=\theta_k$,
    and initializes an empty set of points $\mathcal{Y}$ and a trivial subspace $Z=\Reals^d$. 
    One displacement at a time, the algorithm checks whether the projection onto $Z$ (denoted $\textbf{Proj}_Z$) of a scaled displacement is sufficiently bounded away from zero (that is, the method checks whether the displacement is not sufficiently close to being orthogonal to $Z$). 
    If the projection is sufficiently large, then the displacement is added to the set $\mathcal{Y}$, and the subspace $Z$ is updated to be the null space to the span of the displacement vectors in $\mathcal{Y}$, denoted $\textbf{Null}(\textbf{Span}(\mathcal{Y}))$. 
    After looping over all $p$ points, \Cref{alg.affpoints} returns the union of $x^0$, the $x^i\in\cX$ such that $d^i$ was added to $\mathcal{Y}$, and the set $\{x^0 + \textbf{Basis}(Z)\}$, where $\textbf{Basis}(Z)$ denotes an arbitrary basis for $Z$. 
    In our implementation, and as intended in \cite{wild2008mnh}, all of these projections and null space operations are handled via a QR factorization with insertions, and the final basis for $Z$ is taken as appropriate columns of the orthogonal $Q$ factor. 
    When we call \Cref{alg.affpoints} from \Cref{alg.DFOTR.noise}, the choice of $\cX, c_s,$ and $\Delta_k$ is transparent. We set $\tau = 10^{-5}$. 
    
    \begin{algorithm2e}
         \SetAlgoNlRelativeSize{-4}
    \KwIn{Set of points $\cX = \{x^0, x^1, \dots, x^p\}\subset\Reals^d$, constants $c_s\geq 1$, $\tau\in(0,1/c_s]$, trust region $\Delta_k$.}
    Set $\mathcal{D} := \{d^i = x^i - x^0: i=1,\dots,p\}$.\\
    Initialize $\mathcal{Y} = \emptyset$, $Z = \Reals^d$.\\
    \For{$i=1,\dots,p$}{    
         \If{$\left|\textbf{Proj}_{Z}\left(\frac{1}{c_s\Delta_k}d^i \right)\right| \geq \tau$}{
         $\mathcal{Y} \gets \mathcal{Y}\cup\{d^i\}$\\
         $Z\gets \textbf{Null}(\textbf{Span}(\mathcal{Y}))$\\
         }
    }
    $\cX \gets \{x^0\} \cup \{x^i: d^i\in \mathcal{Y}\} \cup \{x^0 + \textbf{Basis}(Z)\}$.
    \caption{Generating a set of affinely independent points \label{alg.affpoints}}
    \end{algorithm2e}
\end{appendices}

\bibliographystyle{unsrtnat}
\bibliography{vqa}
\end{document}